\documentclass[USenglish, 11pt]{article}
\usepackage[T1]{fontenc}
\usepackage{amsthm,amsmath,amssymb,amsfonts,authblk,hyperref,setspace}
\usepackage{graphicx}
\usepackage{color}
\usepackage{authblk}
\usepackage{booktabs} 
\usepackage{tikz}
\usepackage[ruled]{algorithm2e} 

\SetKwInOut{Input}{Input}
\SetKwInOut{Output}{Output}

\usepackage{cite}
\usepackage{stfloats}
\usepackage{url}

\usepackage{float}
\usepackage{tikz}

\newtheorem{theorem}{Theorem}[section]

\newtheorem{lemma}[theorem]{Lemma}
\newtheorem{claim}[theorem]{Claim}
\newtheorem{proposition}[theorem]{Proposition}

\newtheorem{corollary}[theorem]{Corollary}

\newtheorem{definition}{Definition}

\usepackage{fullpage}

\newcommand{\comment}[1]{}
\newcommand{\tildeO}{\widetilde{O}}

\newcommand{\Z}{\mathbb{Z}}
\newcommand{\N}{\mathbb{N}}

\newcommand{\ed}{\Delta_\textrm{edit}}
\newcommand{\cost}{\textrm{cost}}

\newcommand{\AED}{\textrm{\bf {ED-UB}}}
\newcommand{\UB}{\textrm{\bf {GAP-UB}}}

\newcommand{\ncost}{\textrm{ncost}}
\newcommand{\editd}{d_\textrm{edit}}
\newcommand{\ulamd}{d_\textrm{Ulam}}

\newcommand{\approxfactorB}{840}
\newcommand{\approxfactorA}{1680}
\newcommand{\Ht}{{\widetilde{H}}}
\newcommand{\TE}{{\textrm{\bf SMALL-ED}}}

\newcommand{\Gt}{{\widetilde{G}}}

\newcommand{\DSR}{{\emph{DenseStripRemoval}}}
\newcommand{\SSES}{{\emph{SparseStripExtensionSampling}}}
\newcommand{\CA}{{\emph{CoveringAlgorithm}}}

\begin{document}


\setcounter{page}{0}

\title{Approximating Edit Distance Within Constant Factor in Truly Sub-Quadratic Time}

\author[1]{Diptarka Chakraborty\thanks{diptarka@comp.nus.edu.sg}}
\author[2]{Debarati Das\thanks{debaratix710@gmail.com}}
\author[3]{Elazar Goldenberg\thanks{elazargo@mta.ac.il}}
\author[4]{Michal Kouck{\'{y}}\thanks{koucky@iuuk.mff.cuni.cz}}
\author[5]{Michael Saks\thanks{msaks30@gmail.com}}
\affil[1]{School of Computing, National University of Singapore, Singapore}
\affil[2]{University of Copenhagen, Denmark}
\affil[3]{School of Computer Science, The Academic College of Tel Aviv-Yaffo, Israel}
\affil[4]{Computer Science Institute of Charles University, Prague, Czech Republic}
\affil[5]{Department of Mathematics, Rutgers University, Piscataway, NJ, USA}

\date{}

\maketitle

\begin{abstract}
	
		Edit distance is a measure of similarity of two strings based on
the minimum number of character insertions, deletions, and substitutions required to transform
one string into the other. 
The edit distance can be computed exactly using a dynamic 
programming algorithm that runs in quadratic time. Andoni, Krauthgamer and Onak (2010) gave
a nearly linear time algorithm that approximates edit distance within approximation factor 
$\text{poly}(\log n)$. 

In this paper, we provide an algorithm with running time 
$\tildeO(n^{2-2/7})$ that approximates the edit distance within a constant factor. 
	
\end{abstract}



\thispagestyle{empty}
\newpage

\section{Introduction}
\label{sec:intro}
\noindent
{\bf Exact computation of edit distance.}
The \emph{edit distance} (aka \emph{Levenshtein distance})~\cite{Lev65} between strings $x,y$,
denoted by $\editd(x,y)$, is the minimum number of character insertions, deletions, and substitutions needed to convert $x$ into $y$. It is a widely used distance measure between strings that finds applications in fields such as computational biology, pattern recognition, text processing,  and information retrieval. 
 The problems of efficiently computing $\editd(x,y)$, and of constructing an optimal {\em  alignment} (sequence of operations that converts $x$ to $y$), are of significant interest.

Edit distance can be evaluated exactly in quadratic time  via
dynamic programming (Wagner and Fischer~\cite{WF74}). Landau {\em et al.}~\cite{LMS98}
gave an algorithm that finds an optimal alignment in time $O(n+\editd(x,y)^2)$, improving on a previous $O(n \cdot \editd(x,y))$ algorithm of
Ukkonen~\cite{UKK85}.
Masek and Paterson\cite{MP80} obtained the first (slightly) sub-quadratic  $O(n^2 /\log n)$ time algorithm, and the current asymptotically fastest algorithm (Grabowski~\cite{G16}) runs in time
 $O(n^2\log \log n /\log^2 n)$.    
Backurs and Indyk~\cite{BI15} showed that a {\em truly sub-quadratic algorithm} ($O(n^{2-\delta})$ for some $\delta>0$) would imply a $2^{(1-\gamma)n}$ time algorithm
for CNF-satisfiability, contradicting the Strong Exponential Time Hypothesis (SETH). Abboud et al.~\cite{AHWW16} showed that even shaving an arbitrarily large polylog factor from $n^2$ would have the plausible, but apparently 
hard-to-prove, consequence that NEXP does not have non-uniform ${NC}^1$ circuits. For further ``barrier'' results, see ~\cite{ABW15,BK15}. 

\noindent
{\bf Approximation algorithms.}
There is a long line of work on {\em approximating} edit distance.
The exact $O(n+\editd(x,y)^2)$ time  algorithm of Landau {\em et al.}~\cite{LMS98,Saha14,CGK16} yields a  linear time $\sqrt{n}$-factor approximation.
This approximation factor was improved, first to $n^{3/7}$~\cite{BJKK04}, then to $n^{1/3+o(1)}$~\cite{BES06} and later to $2^{\widetilde{O}(\sqrt{\log n})}$~\cite{AO09}, all with slightly superlinear running time. 
Batu {\em et al.}~\cite{BEKMRRS03} provided an $O(n^{1-\alpha})$-approximation algorithm with running time $O(n^{\max\{\frac{\alpha}{2}, 2\alpha-1\}})$. The strongest result of this type is the $(\log n)^{O(1/\epsilon)}$-factor
approximation (for every $\epsilon>0$) with running time  $n^{1+\epsilon}$ of  Andoni {\em et al.}~\cite{AKO10}.
Abboud and Backurs~\cite{AB17} showed that a truly sub-quadratic deterministic time $1+o(1)$-factor approximation algorithm for  edit distance  would imply new circuit lower bounds.

Independent of our work,
Boroujeni {\em et al.}~\cite{BEGHS18} obtained a truly sub-quadratic {\em quantum} algorithm that provides a constant factor
approximation.
Their latest results ~\cite{BEGHS18a} are a $(3+\epsilon)$-factor with running time  $\tildeO(n^{2-4/21}/\epsilon^{O(1)})$ and
a faster $\tildeO(n^{1.708})$-time  with a larger constant factor approximation.

Andoni and Nguyen ~\cite{AN10} found a randomized algorithm that approximates Ulam distance of two permutations of $\{1,\ldots,n\}$ (edit distance with only insertions and deletions) within a (large) constant factor in time $\tildeO(\sqrt{n}+n/\ulamd(x,y))$, where $\ulamd(x,y)$ is the Ulam distance of the input; this was improved by Naumovitz {\em et al.}~\cite{NSS17} to a $(1+\varepsilon)$-factor approximation (for any  $\varepsilon>0$) with similar running time.

\noindent
{\bf Our results.} We present the first truly sub-quadratic time {\em classical} algorithm  that approximates edit distance within a constant factor.

\begin{theorem}\label{thm:main}
	%
	%
	There is a randomized algorithm \AED{} that on input strings $x,y$ of length $n$ over any alphabet $\Sigma$
outputs an upper bound on $\editd(x,y)$ in time $\tildeO(n^{12/7})$ that, with probability at least 
$1-n^{-5}$, is at most a fixed constant multiple of $\editd(x,y)$.  
\end{theorem}

If the output is $U$, then  the algorithm has implicitly found an alignment of cost at most $U$. The algorithm can be  modified to explicitly output such an alignment. 

The approximation factor proved in this preliminary version is \approxfactorA, 
can be greatly improved by tweaking parameters.  
We believe, but have not proved, that with sufficient care
the algorithm can be modified (with no significant increase in the running time) 
to get $(3+\epsilon)$-approximation.

Theorem~\ref{thm:main} follows from:

\begin{theorem}
\label{thm:UB}
For every $\theta \in [n^{-1/5},1]$,
there is a randomized algorithm $\UB_{\theta}$ that on input strings $x,y$ of length $n$ outputs $u=\UB_{\theta}(x,y)$ such that:
(1) $\editd(x,y) \leq u$ and (2) on any input with $\editd(x,y) \leq \theta n$, 
$u \leq \approxfactorB \theta n$ with probability at least
$1-n^{-7}$.  The running time of $\UB_{\theta}$ is $\tilde{O}(n^{2-2/7}\theta^{4/7})$.
\end{theorem}

The name $\UB_{\theta}$ reflects that this is a "gap algorithm", which distinguishes inputs with
$\editd(x,y) \leq \theta n$ (where the output is at most $\approxfactorB \theta n$), and those with $\editd(x,y)>\approxfactorB \theta n$ (where the 
output is greater than  $\approxfactorB \theta n$).

Theorem~\ref{thm:main} follows via a routine construction of {\AED}  from $\UB_{\theta}$, presented in Section~\ref{sec:sum up}.
The rest of the paper is devoted to proving
Theorem~\ref{thm:UB}.

\noindent
{\bf Further developments.}
Since the initial circulation of our paper, improved algorithms for approximating edit distance have been obtained building on our techniques.
For every $\epsilon>0$, Andoni \cite{And18} presents an algorithms that computes a constant factor approximation 
to edit distance in time $O(n^{\frac{3}{2}+\epsilon})$. Brakensiek and Rubinstein \cite{BR20} and independently Kouck\'y and Saks \cite{KS20}
give for each $\epsilon>0$, an algorithm running in time $O(n^{1+\epsilon})$ that computes a constant factor approximation to edit distance
with additive error $n^{1-\beta}$, for some $\beta>0$. Eventually, Andoni and Nosatzki \cite{AN20} give for each $\epsilon>0$, an algorithm running in 
time $O(n^{1+\epsilon})$ that computes a constant factor approximation to edit distance of any two strings.

\noindent
{\bf The framework of the algorithm.}
We use a standard two-dimensional representation of edit distance.
Visualize $x$ as lying on a horizontal axis and $y$ as lying
on a vertical axis, with horizontal coordinate $i \in \{1,\ldots,n\}$ corresponding to $x_i$ and vertical component
$j$ corresponding to $y_j$.  The width $\mu(I)$ of interval $I \subseteq \{0,1,\ldots,n\}$ is  
$\max(I)-\min(I) = |I|-1$. Also,  $x_I$ denotes the substring
of $x$ indexed by $I\setminus\{\min(I)\}$.  (Note:  $x_{\min(I)}$ is not part of $x_I$, e.g., $x_{\{0,\ldots,n\}}=x$.  
This convention is motivated by Proposition~\ref{prop:def}.)
We refer to $I$ as an {\em $x$-interval} to indicate that it  indexes a substring of $x$, 
and $J$ as a {\em $y$-interval} to indicate that it indexes a substring of $y$. 
A {\em box} is a set $I\times J$ where $I$ is an $x$-interval and  $J$ is a $y$-interval; $I\times J$ corresponds to the substring pair $(x_I,y_J)$. $I \times J$ is a $w$-box if $\mu(I)=\mu(J)=w$. We often abbreviate $\editd(x_I,y_J)$
by $\editd(I,J)$. A {\em decomposition} of an $x$-interval $I$
is a sequence $I_1,\ldots,I_{\ell}$ of subintervals with
$\min(I_1)=\min(I)$, $\max(I_{\ell})=\max(I)$ and for $j \in [\ell-1]$,
$\max(I_j)=\min(I_{j+1})$. If all the subintervals $I_1,\ldots,I_{\ell}$ are of the same width $w$ then it is the $w$-decomposition of $I$
denoted by $\mathcal{I}_{w}(I)$. ($\mathcal{I}_{w}(\{0,\dots,n\})$ is denoted by $\mathcal{I}_{w}$.)

Associated to $x,y$ is a directed graph $G_{x,y}$ with edge costs called
a {\em grid graph} with vertex set $\{0,\ldots,n\} \times \{0,\ldots,n\}$   
and all  edges
of the form  $(i-1,j)\to (i,j)$ ($H$-steps), $(i,j-1) \to (i,j)$ ($V$-steps)
and $(i-1,j-1) \to (i,j)$ ($D$-steps).  Every H-step  or V-step costs 1, and D-steps
cost 1 if $x_i \neq y_j$ and 0 otherwise.  There is a 1-1 correspondence that maps a path  from $(0,0)$ to $(n,n)$ ({\em source-sink path}) to an {\em alignment} from $x$
to $y$, i.e.
a sequence of character deletions, insertions and substitutions that changes $x$ to $y$, where
an H-step  $(i-1,j)\to (i,j)$ means "delete $x_i$", a V-step  $(i,j-1)\to (i,j)$
 means "insert  $y_j$ between $x_i$ and $x_{i+1}$" and a D-step $(i-1,j-1)\to (i,j)$  means
replace $x_i$ by $y_j$, unless they are already equal.  We have:

\begin{proposition}
\label{prop:def}
The cost of an alignment is the sum of edge costs of its associated path $\tau$, $\cost(\tau)$, and  
$\editd(x,y)$ is equal to $\cost(G_{x,y})$, the minimal cost of a path from $(0,0)$ to $(n,n)$.
\end{proposition}

For $I,J \subseteq \{0,\ldots,n\}$, $G_{x,y}(I \times J) \cong G_{x_I,y_J}$ is the grid graph induced
on $I \times J$, and $\editd(I,J) = \cost(G_{x,y}(I \times J))$.
The natural high-level idea of $\UB_{\theta}$  appears (explicitly or implicitly) in previous work. The algorithm has
two phases. First, the {\em covering phase} identifies a set $\mathcal{R}$ of {\em certified boxes} which are 
pairs $(I\times J,\kappa)$, 
where $\kappa$ is
an upper bound on the {\em normalized edit distance} $\ed(x_I,y_J)=\editd(x_I,y_J)/\mu(I)$. ($\ed(I,J)$ is more convenient than $\editd(I,J)$ for
the covering phase.)
Second, the {\em min-cost path phase}, takes input $\mathcal{R}$ and uses a straightforward customized variant of dynamic programming   
to find an upper bound  $U(\mathcal{R})$ on $\editd(x,y)$ in time quasilinear in
 $|\mathcal{R}|$.
The central issue is to ensure that the covering phase outputs  $\mathcal{R}$ that is sufficiently informative so that $U(\mathcal{R}) \leq c \cdot \editd(x,y)$
for constant $c$, while running in sub-quadratic time.

\noindent
{\bf Simplifying assumptions.}
The input strings $x,y$ have equal length $n$. 
 (It is easy to reduce to this case: 
pad the shorter string to the length of the longer using a new symbol.  
The edit distance of the new pair  is between the original edit distance and 
twice the original edit distance. This factor 2 increase in approximation factor can be avoided by generalizing our algorithm to the case $|x| \neq |y|$, but we won't do this here.)  We assume  $n$ is a power of 2 (by padding both strings with a new symbol, which leaves edit distance unchanged).
We assume that $\theta$ is a (negative) integral power of 2.
The algorithm involves integer parameters $w_1,w_2, d$,  all of which are chosen to be powers of 2.

\noindent
{\bf Organization of the paper.} Section~\ref{sec:overview} is a detailed overview of the covering phase
algorithm and its analysis. Section~\ref{sec:covering}
presents the pseudo-code and analysis for the covering phase. Section~\ref{sec:short-path} presents the min-cost path phase algorithm. Section~\ref{sec:sum up} summarizes the full algorithm and 
discusses improvements in running time via recursion.

\section{Covering algorithm: Detailed overview}
\label{sec:overview}

We give a  detailed overview of the covering phase and its
time analysis and proof of correctness, ignoring minor technical details.
The pseudo-code in Section~\ref{sec:covering} corresponds to the
overview, with technical differences mainly to improve 
running time.  We will illustrate the sub-quadratic time analysis with threshold parameter $\theta=n^{-1/50}$ and algorithm parameters
$w_1=n^{1/10}$, $w_2=n^{3/10}$ and $d=n^{1/5}$.  

The covering phase outputs a set $\mathcal{R}$ of certified boxes. 
The goal is that $\mathcal{R}$ includes an  {\em  adequate approximating sequence} for some min-cost 
path $\tau$ in $G_{x,y}$, 
which is a sequence $\sigma$ of certified boxes  $(I_1 \times J_1,\kappa_1),\ldots,(I_\ell \times J_\ell,\kappa_\ell)$ that satisfies:
\begin{enumerate}
\item $I_1,\ldots,I_\ell$ is a decomposition of $\{0,\ldots,n\}$.
\item $I_i \times J_i$ is an {\em adequate cover of $\tau_{i}$}, where $\tau_i=\tau_{I_i}$ denotes
the minimal subpath of $\tau$ whose projection to the $x$-axis is $I_i$,  and adequate cover means that
the (vertical) distance from the start vertex (resp. final vertex) of $\tau_i$ and the lower left (resp. upper right) corner of $G_{x,y}(I_i \times J_i)$,
is at most 
a constant multiple of  $\cost(\tau_i)+\theta \mu(I_i)$.
\item  The sequence $\sigma$ is {\em adequately bounded}, i.e.,
$\sum_i \mu(I_i)\kappa_i \leq c(\cost(\tau)+\theta n)$, for a constant $c$.
\end{enumerate}
This is a slight oversimplification of Definition~\ref{def:approx}.
(See Figure ~\ref{fig:complete} for an illustration.)

The intuition for the second condition is that  $\tau_i$ is "almost" a path
between the lower left and  upper right corners of $I_i \times J_i$.
Now $\tau_i$ might have a vertical extent $J'$ that
is much larger than its horizontal extent $I_i$, in which case it is impossible to place a square $I_i \times J_i$ with
corners close to both endpoints of $\tau_i$.  But in that case, $\tau_i$ has a very high cost (at least $|\mu(J')-\mu(I_i)|$).  The closeness required is
adjusted based on $\cost(\tau_i)$, with relaxed requirements
if $\cost(\tau_i)$ is large.  

The output of the min-cost path phase should satisfy the  requirements of $\UB_{\theta}$. 
Roughly speaking, Lemma \ref{thm:reduction} shows that if the min-cost path phase receives $\mathcal{R}$ 
that contains an adequate approximating sequence to
some min-cost path $\tau$, then it will output an upper bound to $\editd(x,y)$ that is at most $c(\editd(x,y)+\theta n)$ for some $c$.  So that on
input $x,y$ with $\editd(x,y) \leq \theta n$, the output is
at most $2c\theta n$, satisfying the requirements of $\UB_{\theta}$. 
This formalizes the intuition that an adequate approximating sequence
captures enough information to deduce a good bound on $\cost(\tau)$. 

Once and for all, we fix a min-cost path $\tau$.
Our task for the covering phase is that, with high probability, $\mathcal{R}$ includes an adequate approximating
sequence for $\tau$. 

A {\em $\tau$-match} for an $x$-interval $I$ is a $y$-interval $J$ such that 
$I\times J$ is an adequate cover of $\tau_I$. It is easy to show
(Proposition~\ref{prop:cover}) that this implies $\editd(I,J) \leq  O(\cost(\tau_I)+\theta \mu(I))$. A box $I \times J$ is said to be {\em $\tau$-compatible} if $J$ is a $\tau$-match for $I$
and  a box sequence is {\em $\tau$-compatible} if every box  is $\tau$-compatible. 
A $\tau$-compatible certified box sequence  whose distance upper bounds
are (on average) within a constant factor of the actual cost, satisfies the requirements for an adequate
approximating sequence. 
Our cover algorithm will ensure that $\mathcal{R}$ contains such a sequence.

A natural decomposition of $\{0,1,\dots,n\}$ is $\mathcal{I}_{w_1}$, with all parts of width $w_1$
(think of $w_1$ as a power of 2 that is roughly $n^{1/10}$) so $\ell=n/w_1$ and  $I_j= \{(j-1)w_1, \cdots, j w_1\}$.
An interval of width $w_1$  is {\em $\theta$-aligned}
if its upper and lower endpoints are both multiples of $\theta w_1$ (which we require to be an integral power of 2). 
We restrict attention to
$x$-intervals in $\mathcal{I}_{w_1}$, called {\em $x$-candidates} and $\theta$-aligned $y$-intervals of width $w_1$
called {\em $y$-candidates}.
It can be shown (see Proposition~\ref{prop:align-box}) that an $x$-interval
$I$ always has a $\tau$-match $J$ that is $\theta$-aligned.  
For each $x$-candidate $I$, designate one such $\tau$-match as the {\em canonical $\tau$-match}, $J^{\tau}(I)$ for $I$, and $I \times J^{\tau}(I)$ is the {\em canonical $\tau$-compatible box for $I$}.

The na\"{i}ve approach to building $\mathcal{R}$ is to include certified boxes for all choices of $J$ to guarantee a $\tau$-match for each $I_j$, i.e., to consider each   
($x$-candidate, $y$-candidate)-pair $(I,J)$, compute its edit distance in
time $O(w_1^2)$,
and include the certified box $(I \times J,\ed(I,J))$ in $\mathcal{R}$.  
There are $\frac{n}{w_1}\cdot \frac{n}{\theta w_1}$ boxes, so the time for all edit distance computations is
$O(\frac{n^2}{\theta})$, which is worse than quadratic. (The factor $\frac{1}{\theta}$ can be avoided by standard techniques, but this is not significant to the quest for a sub-quadratic algorithm, so we defer this until the next section.)
Note  that $|\mathcal{R}|$ is $\frac{n^2}{\theta (w_1)^2}$ (which is $n^{1.82}$ for our parameters) 
so at least the min-cost path phase (which runs in time quasi-linear in $\mathcal{R}$) is truly sub-quadratic.

Two natural goals that will improve the running time are:  
(1) Reduce the amortized time per box needed to certify boxes significantly below $(w_1)^2$ and 
(2) Reduce the total number of certified boxes created significantly below $\frac{n^2}{\theta (w_1)^2 }$.
Neither goal is always achievable, and our covering algorithm combines them. 
In independent work  ~\cite{BEGHS18,BEGHS18a}, versions of these two goals are combined, where the second goal is accomplished via Grover search, thus yielding
a constant factor sub-quadratic time quantum approximation algorithm.
(1) builds on ideas from \cite{CGK18} exploiting periodicity of the strings.

{\bf Reducing amortized time for certifying boxes: the {\em dense case} algorithm.}
We aim to reduce the amortized time per certified box to be much smaller than $(w_1)^2$.  We divide our search for certified boxes into iterations $i \in \{0,\ldots,\log n\}$.
For iteration $i$, with $\epsilon_i=2^{-i}$, our goal is that for all candidate pairs $I,J$ with $\ed(I,J) \leq \epsilon_i$, we include the certified box $(I \times J,c \epsilon_i)$ for a  fixed constant $c$.  If we succeed,
then for each $I_j$ and its canonical $\tau$-match $J^{\tau}(I_j)$, and for
the largest index $i$ for which $\ed(I_j,J^{\tau}(I_j)) \leq \epsilon_i$,
iteration $i$ will certify $(I_j \times J^{\tau}(I_j),\kappa_j)$ with $\kappa_j \leq c \epsilon_ i\leq 2c \ed(I_j,J^{\tau}(I_j))$, as needed. 

For a string $z$ of size $w_1$, let $\mathcal{H}(z,\rho)$ be the set of $x$-candidates
$I$ with $\ed(z,x_I) \leq \rho$ and  $\mathcal{V}(z,\rho)$ be the
set of $y$-candidates $J$ with $\ed(z,y_J) \leq \rho$. In iteration $i$,
for each $x$-candidate $I$, we will specify a set $\mathcal{Q}_i(I)$ of $y$-candidates
that includes $\mathcal{V}(x_I,\epsilon_i)$ and is contained in $\mathcal{V}(x_I,5 \epsilon_i)$.
The set of certified boxes $(I \times J,5\epsilon_i)$ for all 
$x$-candidates $I$ 
and $J \in \mathcal{Q}_i(I)$ satisfies the goal of iteration $i$.

Iteration $i$ proceeds in rounds. 
At the beginning of the $i$-th iteration we mark all $x$-candidates as {\em unfulfilled}. 
In each round we select an $x$-candidate $I$, called the {\em pivot}, 
which is unfulfilled. 
Compute $\ed(x_I,y_J)$ for all $y$-candidates $J$ and
$\ed(x_{I},x_{I'})$ for all $x$-candidates $I'$; these determine
$\mathcal{H}(x_I,\rho)$ and $\mathcal{V}(x_I,\rho)$ for any $\rho$.
For all $I' \in \mathcal{H}(x_I,2\epsilon_i)$, set 
 $\mathcal{Q}_i(I')=\mathcal{V}(x_I,3\epsilon_i)$.
By the triangle inequality, for each $I' \in \mathcal{H}(x_I,2\epsilon_i)$, 
$\mathcal{V}(x_I,3\epsilon_i)$ includes $\mathcal{V}(x_{I'},\epsilon_i)$
and is contained in $\mathcal{V}(x_{I'},5\epsilon_i)$ so we can certify all the boxes
with upper bound $5\epsilon_i$.
Mark  intervals in $\mathcal{H}(x_I,2\epsilon_i)$ as {\em fulfilled} and proceed
to the next round, choosing a new pivot from among the unfulfilled $x$-candidates. 

The number of certified boxes produced  in a round is $|\mathcal{H}(x_I,2\epsilon_i)|\times
|\mathcal{V}(x_I,3\epsilon_i)|$.  If this is much larger than $O(\frac{n}{\theta w_1})$, the
number of edit
distance computations, then we have significantly reduced amortized
time  per certified box. (For example, in the trivial case $i=0$, every candidate box will be certified
in a single round.) But
in worst case, there are $\frac{n}{w_1}$ rounds each requiring $\Omega(\frac{nw_1}{\theta})$ time, 
for an unacceptable total time
$\Theta(n^2/\theta)$.

Here is a situation where the number of rounds is much less than $\frac{n}{w_1}$. 
Since any two pivots are necessarily greater than $2\epsilon_i$ apart,
the sets $\mathcal{V}(x_I,\epsilon_i)$ for distinct pivots are disjoint.
Now for some parameter $d$ (think of $d=n^{1/5}$)
an $x$-candidate is {\em $d$-dense} for $\epsilon_i$
if $|\mathcal{V}(x_I,\epsilon_i)| \geq d$, i.e., $x_I$ is $\epsilon_i$-close
in edit distance to at least $d$  $y$-candidates; it is $d$-sparse otherwise.
If we manage to select a $d$-dense pivot $I$ in each round, 
then the number of rounds is $O(\frac{n}{w_1d\theta})$ and the overall time will be 
$\Theta(\frac{n^2}{d\theta^2})$.
For our parameters this is $\Theta(n^{1.84})$.
But there's no reason to expect that we'll only choose dense pivots; indeed there need
not be any dense pivot.  

Let's modify the process a bit.  When choosing potential pivot $I$, 
first test whether or not it is (approximately) $d$-dense.  This can be done with high probability, by randomly
sampling $\tilde{\Theta}(\frac{n}{\theta w_1 d})$ $y$-candidates and finding the fraction of the sample
that are within $\epsilon_i$ of $x_I$.  If this fraction is less than $\frac{\theta w_1 d}{2n}$ then $I$ is {\em declared sparse}
and abandoned as a pivot; otherwise $I$ is {\em declared dense}, and used as a pivot.   With high probability, all $d$-dense intervals that are tested are declared dense, and all tested intervals that are not
$d/4$-dense are declared sparse, so we assume this is the case.  Then all
pivots are  processed (as above) in time $O(\frac{n^2}{d \theta^2})$ (under our sample parameters: $O(n^{1.84})$). 
We pay 
$\tilde{O}(\frac{n}{w_1d\theta}) (w_1)^2$ to test each potential pivot
(at most $\frac{n}{w_1}$ of them)  so the overall time to test potential
pivots is $\tilde{O}(\frac{n^2}{d\theta})$ (with our parameters: $\tilde{O}(n^{1.82})$).

Each iteration $i$  (with different $\epsilon_i$) splits $x$-candidates into
two sets, $\mathcal{S}_i$ of intervals that are declared sparse, and all of the rest
for which we have found the desired set $\mathcal{Q}_i(I)$.     
With high probability every interval
in $\mathcal{S}_i$
is indeed $d$-sparse, but a sparse interval need not belong to $\mathcal{S}_i$, since it may belong to $\mathcal{H}(x_I,2\epsilon_i)$ for some selected
pivot $I$.

For every $x$-candidate  $I \not\in \mathcal{S}_i$ we have met the goal for the iteration.
If $\mathcal{S}_i$ is very small for all iterations (with respect to $\epsilon_i$), then the  set of certified boxes
will suffice for the min-cost path algorithm to output a good approximation.
But if $\mathcal{S}_i$ is not small, another approach is needed.

\noindent
{\bf Reducing the number of candidates explored: the {\em diagonal extension algorithm}.}
For each $x$-candidate $I$, although it suffices to certify the single box $(I,J^{\tau}(I))$ with a good
upper bound, since $\tau$ is unknown,
the exhaustive and dense case approaches both include certified boxes for all $y$-candidates $J$.
The potential savings in the dense case approach comes from certifying many boxes simultaneously
using a relatively small number of edit distance computations.

Here's another approach: for each $x$-candidate $I$ try to quickly identify a relatively small subset $\mathcal{Y}(I)$ of
$y$-candidates that is guaranteed to include $J^{\tau}(I)$.  If we succeed, then the  number of boxes we certify  is significantly reduced, and even paying quadratic time per certified box, 
 we will have a sub-quadratic algorithm. 

We need the notion of {\em diagonal extension} of a box.  The {\em main diagonal} of box $I \times J$, 
is the segment joining the lower left and upper right corners. The square box $I' \times J'$ is a {\em diagonal
extension} of a square subbox $I \times J$ if the main diagonal of $I \times J$ is a subsegment
of the main diagonal of $I' \times J'$ (see Definition~\ref{def:DE}).
Given square box $I \times J$ and $I \subset I'$ the {\em diagonal
extension of $I \times J$ with respect to $I'$} is the unique diagonal extension of $I \times J$ having $x$-interval $I'$.
The key observation (Proposition 
~\ref{prop:extension}) is: if $I \times J$ is an adequate
cover of $\tau_{I}$ 
then any diagonal extension $I' \times J'$ is an adequate cover of $\tau_{I'}$.    

Now 
let $w_1, w_2$ be two numbers with $w_1|w_2$ and $w_2|n$.  (Think of $w_1=n^{1/10}$ and 
$w_2=n^{3/10}$.)   
We use the decomposition $\mathcal{I}_{w_2}$ of $\{0,\ldots,n\}$ into intervals of width $w_2$. The set of $y$-candidates consists of $\theta$-aligned vertical intervals of width $w_2$ and has size
$\frac{n}{\theta w_2}$ . 
To identify a small set of potential matches for $I' \in \mathcal{I}_{w_2}$, we will identify a  set (of size much smaller than $\frac{n}{w_2}$) of $w_1$-boxes
$\mathcal{B}(I')$ having $x$-interval in $\mathcal{I}_{w_1}(I')$ (the decomposition of $I'$ into width $w_1$ intervals).  For each box in $\mathcal{B}(I')$
we determine the diagonal extension $I' \times J'$ with respect to $I'$, compute $\kappa=\ed(I',J')$ and certify
$(I' \times J',\kappa)$.
Our hope is that $\mathcal{B}(I')$
includes a $\tau$-compatible $w_1$-box $I'' \times J^{\tau}(I'')$. If so then
the observation above implies that its 
diagonal extension provides 
an adequate cover for $\tau_{I'}$.

Here's how to build $\mathcal{B}(I')$:
Randomly select a polylog$(n)$ size set $\mathcal{H}(I')$ of width $w_1$ intervals
from $\mathcal{I}_{w_1}(I')$.   For each $I'' \in \mathcal{H}(I')$ 
compute $\ed(I'',J'')$ for each
$y$-candidate $J''$, and let $\mathcal{J}(I'')$ consist of the $d$  candidates $J''$ 
with smallest edit distance to $I''$. Here $d$ is a parameter; think of $d=n^{1/5}$ as before.  
$\mathcal{B}(I')$ consists of all
$I'' \times J''$ where $I'' \in \mathcal{H}(I')$ and $J'' \in \mathcal{J}(I'')$.

To bound running time: Each $I' \in \mathcal{I}_{w_2}$ requires
$\tilde{O}(\frac{n}{\theta w_1})$ width-$w_1$ $\ed()$ computations, taking time $\tilde{O}(\frac{nw_1}{\theta})$. Diagonal extension step requires $\tilde{O}(d)$ width-$w_2$ $\ed()$ computations, for time $\tilde{O}(dw_2^2)$.
Summing over $\frac{n}{w_2}$ choices for $I'$ gives time $\tilde{O}(n^2 \frac{w_1}{\theta w_2} + n dw_2)$ (with our parameters: $\tilde{O}(n^{1.82})$).

Why should  $\mathcal{B}(I')$ include
a box that is an adequate approximation to $\tau_{I'}$?  The intuition behind the 
choice of $\mathcal{B}(I')$
is that an adequate cover for $\tau_{I'}$ 
should typically be among the cheapest boxes of the form $I' \times J'$, and
if $I' \times J'$ is cheap then for a randomly chosen $w_1$-subinterval
$I''$, we should also have $I'' \times J^{\tau}(I'')$ is among the cheapest
boxes for $I''$.

Clearly this intuition is faulty: $I'$ may have many inexpensive matches $J'$
such that $I' \times J'$ is far from $\tau_{I'}$, which may all be 
much cheaper than the match we are looking for. In this bad situation, there are many  $y$-intervals $J'$
such that $\ed(I',J')$ is smaller than the match we are looking for, and this is  reminiscent of the {\em good} situation for
the dense case algorithm, where we hope that $I'$ has lots of close matches.  This suggests combining
the two approaches, and leads to our full covering algorithm. 

\noindent
{\bf The full covering algorithm.}
Given the dense case and diagonal extension algorithms, the full covering algorithm is easy to
describe.  
The parameters $w_1,w_2,d$ are as above.  
We iterate over $i \in \{0,\ldots,\log n\}$ with 
$\epsilon_i=2^{-i}$.  In iteration $i$, we first run the dense case algorithm, and let
$\mathcal{S}_i$ be the set of intervals declared sparse.
Then run the diagonal extension algorithm described earlier (with
small modifications): For each
width $w_2$ interval $I'$, select $\mathcal{H}(I')=\mathcal{H}_i(I')$ to consist of
$\theta(\log^2 n)$ independent random selections from $\mathcal{S}_i$. 
For each $I'' \in \mathcal{H}_i(I')$, find the set of vertical
candidates $J''$ for which $\ed(I'',J'') \leq \epsilon_i$.  Since $I''$
is (almost certainly) $d$-sparse, the number of such $J''$ is at most $d$.
Proceeding as in the diagonal extension algorithm, we produce
a set $\mathcal{P}_i(I')$  of $\tilde{O}(d)$ certified $w_2$-boxes with $x$-interval $I'$.  
Let $\mathcal{R}_D$ (resp. $\mathcal{R}_E$) be the set
of all certified boxes produced by the dense case iterations, resp. diagonal extension iterations.
The output is $\mathcal{R}=\mathcal{R}_D \cup \mathcal{R}_E$.
(See Figure ~\ref{fig:complete} for an illustration of the output $\mathcal{R}$.)

The running time is the sum of the running times of the dense case and diagonal extension algorithms, as analyzed above. Later, we will give a more precise analysis of the running time
for the pseudo-code.

To finish this extended overview, we sketch the argument that $\mathcal{R}$ satisfies the covering phase requirements.

\begin{claim}
\label{claim:main overview}
Let $I'$ be an interval in the $w_2$-decomposition. Either
(1) the output of the dense case algorithm includes
a sequence of certified $w_1$-boxes that adequately approximates the subpath $\tau_{I'}$, or (2) with high probability the output of the sparse case algorithm includes a single $w_2$-box that adequately approximates 
$\tau_{I'}$.
\end{claim}
(This claim is formalized in Claim~\ref{claim:main}.)
Stitching together the subpaths for all $I'$ implies that $\mathcal{R}$ contains a sequence of certified boxes that adequately approximates $\tau$.

To prove the claim,  we establish a sufficient condition for each of the two
conclusions and show that if the sufficient condition for the second conclusion fails, then
the sufficient condition for the first holds.

\begin{center}
   \begin{figure}[ht]
\centering
\begin{tikzpicture}[scale=0.25,shorten >=1mm,>=latex]
 \tikzstyle gridlines=[color=black!20,very thin]
 
 

  \draw[color=black,fill=gray!30] (4,19.1) rectangle (10,25.1);
  \draw[color=black,fill=yellow!20] (7,22.1) rectangle (8.5,23.6);

  \draw[color=black,fill=gray!30] (4,10) rectangle (10,16);
  \draw[color=black,fill=yellow!20] (5.5,10.7) rectangle (7,12.2);

  \draw[color=black,fill=gray!30] (16,28) rectangle (22,34);
   \draw[color=black,fill=yellow!20] (16,28) rectangle (17.5,29.5);

   \draw[color=black,fill=gray!25] (10,25) rectangle (16,31);
  \draw[color=black,fill=yellow!20] (11.5,26.5) rectangle (13,28);

 \draw[color=black,fill=blue!25] (-2,10) rectangle (-.5,11.5);
 \draw[color=black,fill=blue!25] (-2,14.7) rectangle (-.5,16.2);
 \draw[color=black,fill=blue!25] (-2,24.55) rectangle (-.5,26.05);
 \draw[color=black,fill=blue!25] (-2,31) rectangle (-.5,32.5);
 \draw[color=black,fill=blue!25] (-2,11.5) rectangle (-.5,13);
 
 \draw[color=black,fill=blue!25] (-.5,11.5) rectangle (1,13);.
  \draw[color=black,fill=blue!25] (-.5,10) rectangle (1,11.5);
   \draw[color=black,fill=blue!25] (-.5,14.7) rectangle (1,16.2);
   \draw[color=black,fill=blue!25] (-.5,24.55) rectangle (1,26.05);
   \draw[color=black,fill=blue!25] (-.5,31) rectangle (1,32.5);
 
 \draw[color=black,fill=blue!25] (2.5,10) rectangle (4,11.5);
 \draw[color=black,fill=blue!25] (2.5,14.7) rectangle (4,16.2);
 \draw[color=black,fill=blue!25] (2.5,24.55) rectangle (4,26.05);
 \draw[color=black,fill=blue!25] (2.5,31) rectangle (4,32.5);
 \draw[color=black,fill=blue!25] (2.5,11.5) rectangle (4,13);
   
 \draw[color=black,fill=blue!25] (10,10) rectangle (11.5,11.5);
 \draw[color=black,fill=blue!25] (10,14.7) rectangle (11.5,16.2);
 \draw[color=black,fill=blue!25] (10,24.55) rectangle (11.5,26.05);
 \draw[color=black,fill=blue!25] (10,31) rectangle (11.5,32.5);
 \draw[color=black,fill=blue!25] (10,11.5) rectangle (11.5,13);
     
 \draw[color=black,fill=blue!25] (19,10) rectangle (20.5,11.5);
 \draw[color=black,fill=blue!25] (19,14.7) rectangle (20.5,16.2);
 \draw[color=black,fill=blue!25] (19,24.55) rectangle (20.5,26.05);
 \draw[color=black,fill=blue!25] (19,31) rectangle (20.5,32.5);
 \draw[color=black,fill=blue!25] (19,11.5) rectangle (20.5,13);

 \draw[color=black,fill=blue!25] (13,11.5) rectangle (14.5,13);
 \draw[color=black,fill=blue!25] (13,20) rectangle (14.5,21.5);
 \draw[color=black,fill=blue!25] (13,18.5) rectangle (14.5,20);
 \draw[color=black,fill=blue!25] (13,29.5) rectangle (14.5,31);
 
 \draw[color=black,fill=blue!25] (14.5,11.5) rectangle (16,13);
 \draw[color=black,fill=blue!25] (14.5,20) rectangle (16,21.5);
 \draw[color=black,fill=blue!25] (14.5,18.5) rectangle (16,20);
 \draw[color=black,fill=blue!25] (14.5,29.5) rectangle (16,31);

 \draw[color=black,fill=blue!25] (20.5,13) rectangle (22,14.5);
 \draw[color=black,fill=blue!25] (20.5,23) rectangle (22,24.5);
 \draw[color=black,fill=blue!25] (20.5,32.5) rectangle (22,34);
 \draw[color=black,fill=blue!25] (20.5,17) rectangle (22,18.5);

 \draw[color=black,fill=blue!25] (1,13) rectangle (2.5,14.5);
 \draw[color=black,fill=blue!25] (1,23) rectangle (2.5,24.5);
 \draw[color=black,fill=blue!25] (1,32.5) rectangle (2.5,34);
 \draw[color=black,fill=blue!25] (1,17) rectangle (2.5,18.5);

 \draw[color=black] (-2,10) rectangle (22,34);
  
  \draw[color=black!90,dashed] (4,8.5) -- (4,35.5);
  \draw[color=black!90,dashed] (10,8.5) -- (10,35.5);
  \draw[color=black!90,dashed] (16,8.5) -- (16,35.5);

  \draw[color=black!90,dotted] (-.5,10) -- (-.5,34.2);
    \draw[color=black!90,dotted] (1,10) -- (1,34.2);
    \draw[color=black!90,dotted] (2.5,10) -- (2.5,34.2);
   
   \draw[color=black!90,thick,dotted] (5.5,10) -- (5.5,34.2);
   \draw[color=black!90,thick,dotted] (7,10) -- (7,34.2);
   \draw[color=black!90,thick,dotted] (8.5,10) -- (8.5,34.2);
   
   \draw[color=black!90,thick,dotted] (11.5,10) -- (11.5,34.2);
   \draw[color=black!90,thick,dotted] (13,10) -- (13,34.2);
   \draw[color=black!90,thick,dotted] (14.5,10) -- (14.5,34.2);
   
   \draw[color=black!90,thick,dotted] (17.5,10) -- (17.5,34.2);
   \draw[color=black!90,thick,dotted] (19,10) -- (19,34.2);
   \draw[color=black!90,thick,dotted] (20.5,10) -- (20.5,34.2);

\draw[color=red,very thick] (-2,10) -- (2.2,14.2);
 \draw[color=red,very thick] (2,14) -- (2,15.2);
 \draw[color=red,very thick] (2,15) -- (3,15);
 \draw[color=red,very thick] (2.8,15) -- (5,17.2);
 \draw[color=red,very thick] (4.8,17) -- (4.8,19.5);
 \draw[color=red,very thick] (4.8,19.3) -- (6.7,21.2);
 \draw[color=red,very thick] (6.5,21) -- (6.5,21.9);
 \draw[color=red,very thick] (6.5,21.7) -- (9.6,24.8);
 \draw[color=red,very thick] (9.4,24.7) -- (10.3,24.7);
 \draw[color=red,very thick] (10.1,24.7) -- (11.7,26.2);
 \draw[color=red,very thick] (11.5,26) -- (11.5,26.7);
 \draw[color=red,very thick] (11.5,26.5) -- (13.2,28.2);
 \draw[color=red,very thick] (13,28) -- (16.2,28);
 \draw[color=red,very thick] (16,28) -- (22.2,34.2);

 \draw[color=black!90,dashed] (10,25) -- (16,31);
  \draw[color=black!90,dashed] (4,10) -- (10,16);
   \draw[color=black!90,dashed] (4,19.1) -- (10,25.1);
   \draw[color=black!90,dashed] (16,28) -- (22,34);

 \node at (-3,21) {$J$};
 \node at (8,35) {$I$};

 \draw [black,<->,thin] (9.8,34.5) -- (16.2,34.5);
 \node at (13.3,35.3) {$w_2$};
 \draw [black,<->,thin] (17.3,34.5) -- (19.2,34.5);
 \node at (18.25,35.3) {$w_1$};
 
 \end{tikzpicture}
 \caption{Illustration of the Covering Algorithm: Blue boxes are low cost boxes in dense $w_1$-strips, while the yellow ones are in sparse $w_1$-strips. The red line corresponds to the path $\tau$ that we are trying to cover. In each $w_2$-strip, $\tau$ is covered by either a collection of many $w_1$-boxes or it is covered by a diagonal extension of a low cost $w_1$-box. The various boxes might overlap vertically
which is not shown in the picture.}  
   \label{fig:complete}   
 \end{figure}

\end{center}

Let $\mathcal{I}'$ denote the $w_1$-decomposition $\mathcal{I}_{w_1}(I')$ of $I'$.
Every interval $I'' \in \mathcal{I}'$ has a $\theta$-aligned $\tau$-match $J^{\tau}(I'')$.
It will be shown (see Proposition~\ref{prop:align-box}), that $\ed(I'',J^{\tau}(I'')) \leq 2\frac{\cost(\tau_{I''})}{\mu(I'')}
+\theta$.  Let $u(I'')$ denote this upper bound.
Consider the first alternative in the claim. During the dense case iteration $i=0$,
every interval is declared dense, and 
$(I''\times J^{\tau}(I''),5)$ is in $\mathcal{R}_D$ for all $I''$.   To get an  adequate approximation, we
try to show that later iterations provide much better upper bounds on these boxes,
i.e., $(I''\times J^{\tau}(I''),\gamma(I'')) \in \mathcal{R}_D$ for a small enough value of $\gamma(I'')$.
By definition
of adequate approximation, it is enough that
$\sum_{I'' \in \mathcal{I}'} \gamma(I'') \leq c \sum_{I'' \in \mathcal{I'}}u(I'')$, for some $c$.
Let $t(I'')$ be the last (largest) iteration for which  $\epsilon_{t(I'')} \geq u(I'')$ and $I'' \not\in \mathcal{S}_{t(I'')}$
(which is well defined since $\mathcal{S}_0=\emptyset$).
Let $b(I'')=\epsilon_{t(I'')}$. Since $b(I'') \geq u(I'') \geq \ed(I'',J^{\tau}(I''))$,
the box
$(I'' \times J^{\tau}(I''),5b(I''))$ is certified.   The collection $\{(I'' \times J^{\tau}(I''),5b(I''))\}$ is a   sequence of certified
boxes that satisfies the first two conditions for
an adequate approximation of $\tau$.  The third condition will follow if:

\begin{equation}
\label{1st cond succeeds} 
\sum_{I'' \in \mathcal{I}'} 5 b(I'')  \leq c \sum_{I'' \in \mathcal{I'}} u(I'')
\end{equation}
so this is sufficient to imply the first condition of the claim.

Next consider what we need for the second alternative to hold.  Let $\mathcal{S}_i(I')$ be the set of intervals declared
sparse in iteration $i$.  An interval $I'' \in \mathcal{S}_i(I')$ is a {\em winner} (for iteration $i$) if $\ed(I'',J^{\tau}(I'')) \le \epsilon_i$, and $\mathcal{W}_i(I')$ is the set of winners. In iteration $i$ of the diagonal extension algorithm, we sample $\theta(\log^2 n)$ elements of $\mathcal{S}_i(I')$.
If for at least one iteration $i$ our sample includes a winner $I''$ then the second condition of the claim will hold:
$I'' \times J^{\tau}(I'')$ is extended diagonally to a $w_2$-box, and by the diagonal extension property, the extension
is an adequate cover of $\tau_{I'}$, which we will certify with its exact edit distance.  

Observe that if there exists $i$ such that $|\mathcal{W}_i(I')| \ge |\mathcal{S}_i(I')-\mathcal{W}_i(I')|$
then, except with negligible probability, 
during the $i$-th iteration our sample includes a winner $I''$ 
(as for each sample the probability of being a winner is at least $1/2$).
Thus for the second alternative to fail with nonnegligible probability:

\begin{equation}
\label{2nd cond fails}
\text{For all $i$}, |\mathcal{W}_i(I')| < |\mathcal{S}_i(I')-\mathcal{W}_i(I')|.
\end{equation}

We argue next that if the failure condition (\ref{2nd cond fails}) holds,
then the success condition (\ref{1st cond succeeds}) holds. 
Multiply (\ref{2nd cond fails}) by $\epsilon_i$ and sum on $i$ to get:

\begin{equation}
\label{2nd cond fails-2}
\sum_{I'' \in \mathcal{I}'} \sum_{i:I'' \in \mathcal{W}_i(I')} \epsilon_i < 
\sum_{I'' \in \mathcal{I}'} \sum_{i:I''\in \mathcal{S}_i(I')-\mathcal{W}_i(I')} \epsilon_i.
\end{equation}

For a given interval $I'' \in \mathcal{I}_{w_1}(I')$, consider the iterations $i$ for which
$I'' \in \mathcal{W}_i(I')$ and those for which $I'' \in \mathcal{S}_i(I')-\mathcal{W}_i(I')$.
First of all if $\epsilon_i \geq u(I'')$ and $I'' \in \mathcal{S}_i(I')$ then since $\ed(I'',J^{\tau}(I'')) \leq u(I'') \leq \epsilon_i$ we conclude $I'' \in \mathcal{W}_i(I')$.  So $I'' \in \mathcal{S}_i(I')-\mathcal{W}_i(I')$ implies that $\epsilon_i < u(I'')$, so
the inner sum of the right side of (\ref{2nd cond fails-2}) is at most $2u(I'')$ (by summing a geometric series).

Furthermore, for $i$ with $u(I'') \leq \epsilon_i < b(I'')$, $I'' \in \mathcal{S}_i$ by the choice of $t(I'')$.
Either $b(I'')/2 \leq u(I'')$ or 
$u(I'')<b(I'')/2$. The latter implies $I'' \in \mathcal{W}_{t(I'')+1}(I')$, and then $b(I'')/2$ is
upper bounded by the inner sum on the left of (\ref{2nd cond fails-2}).  Therefore:

\begin{eqnarray*}
\sum_{I''} b(I'') & \leq &  \sum_{I''} \left( 2u(I'')+\sum_{i:I'' \in \mathcal{W}_i(I')} 2 \epsilon_i\right) \\
&<&\sum_{I''} \left( 2u(I'')+2 \sum_{i:I'' \in \mathcal{S}_i(I') - \mathcal{W}_i(I') }\epsilon_i \right)\\
& \leq & 6 \sum_{I''} u(I''),
\end{eqnarray*}
as required for (\ref{1st cond succeeds}).

This completes the overview of the covering algorithm.

\section{Covering Algorithm: pseudo-code and analysis}
\label{sec:covering}

The pseudo-code  consists of \CA{} which
calls procedures \DSR{} (the dense case algorithm) and 
\SSES{} (the diagonal extension algorithm). The technical differences between the pseudo-code and the informal description, are mainly to improve analysis of the running time.

\subsection{Pseudo-code}
The parameters of \CA{} are as described in the overview: $x,y$ are  input
strings of length $n$, $\theta$ comes from $\UB_{\theta}${}, $w_1<w_2<n$ and $d<n$ are integral powers of 2, as are
the auxiliary input parameters. The output is a set $\mathcal{R}$ of certified boxes.
The algorithm uses  global constants $c_0\ge 0$ and $c_1\ge 640$, where the former one is needed for Proposition \ref{prop:sampling}.

We use a subroutine \TE{} which takes strings $z_1,z_2$ of length $w$ and parameter $\kappa$ and outputs $\infty$ if $\ed(z_1,z_2) > \kappa$ and otherwise outputs $\ed(z_1,z_2)$.   The  algorithm of \cite{UKK85} implements \TE{} in time $O(\kappa w^2)$.

One  technical difference from the overview, is that
the pseudo-code saves time by restricting the search for
certified boxes  to
a portion of the grid close to the main diagonal.  Recall that $\UB_{\theta}$ has two requirements,
that the output upper bounds $\editd(x,y)$ (which will be guaranteed by the
requirement that
$\mathcal{R}$ contains no falsely certified boxes), and that
if  $\editd(x,y) \leq \theta n$,
the output is at most $c \theta n$ for some constant $c$.  We therefore
design our algorithm
assuming $\editd(x,y) \leq \theta n$, in which case every min-cost
$G_{x,y}$-path $\tau$ consists entirely of points $(i,j)$
within $\frac{\theta}{2}n$ steps from 
the main diagonal, i.e. $|i-j| \leq \frac{\theta}{2} n$. So we restrict
our search for certified boxes as follows: set
$m=\frac{1}{4} \theta n$, and consider the $\frac{n}{m}$ overlapping
equally spaced
boxes of width $8m=2\theta n$ lying along the main diagonal. Together these boxes cover all points within $\theta n$ of the main diagonal.  

The algorithm of the overview is executed separately on each of these $n/m$ boxes.   Within each of these executions, we iterate over  $i \in \{0,\ldots,\log \frac{1}{\theta}\}$ (rather than
$\{0,\ldots, \log n\}$ as in the overview).  In each iteration we apply the
dense case algorithm and the diagonal extension algorithm as in the overview.
The output is the union over all $n/m$ boxes and all iterations, of the
boxes produced.

In the procedures \DSR{} and \SSES{}, the input $G$
is an induced grid graph corresponding to a box $I_G \times J_G$, as described in the
"framework" part of Section
~\ref{sec:intro}.
The procedure \DSR{} on input $G$, sets $\mathcal{T}$ to be
the $w_1$-decomposition of $I_G$ (the $x$-candidates) and 
$\mathcal{B}$ to be the set of $\frac{\epsilon_i}{8}$-aligned
$y$-candidates. As in the overview, the dense case algorithm
produces a set of certified boxes (called $\mathcal{R}_1$ in the pseudo-code) and a
set $\mathcal{S}$ of intervals declared sparse.
\SSES{} is invoked if  $\mathcal{S} \neq \emptyset$ and iterates
over all $x$-intervals $I'$ in the decomposition $\mathcal{I}_{w_2}(I_G)$.
The algorithm skips $I'$ if $\mathcal{S}$
contains no subset of $I'$, and otherwise selects a sample
$\mathcal{H}$ of $\theta(\log^2 n)$ subintervals of $I'$ from $\mathcal{S}$.
For each sample interval $I''$ it finds the vertical candidates $J''$ for which $\ed(I'',J'')\le \epsilon_i$, does a diagonal extension to $I'$ and certifies each box with an exact edit distance computation.

There are  a few parameter changes from the overview that provide some improvement in the
time analysis:
During each iteration $i$, rather than taking our vertical candidates to be from a 
$\theta$-aligned grid, we can afford a grid that is $\epsilon_i/8$-aligned which is coarser for $\epsilon_i \gg \theta$. Also, the local parameter $d$ in \DSR{} and \SSES{} is set to $d/\epsilon_i$ during iteration $i$.

There is one counterintuitive quirk in \SSES{}: each certified box 
is replicated $O(\log n)$ times with higher
distance bounds.  This is permissible (increasing the distance bound cannot decertify a box),
but seems silly (why add the same box
with a higher distance bound?).  This is just a convenient technical device
to ensure that
the second phase min-cost path
algorithm gives a good approximation (using adjacent boxes that overlap vertically).

\begin{algorithm}
 
   \Input{Strings $x,y$ of length $n$, $w_1,w_2,d\in [n]$, $w_1<w_2 < \theta n/4$, and $\theta \in [0,1]$. $n,w_1,w_2,\theta$ are powers of 2.}
   
   \Output{A set $\mathcal{R}$ of certified boxes in $G$.}
   
   \vspace{1mm}
   \hrule\vspace{1mm}
 
   Initialization: $G=G_{x,y}$, $\mathcal{R}_D=\mathcal{R}_E=\emptyset$.
   
   Let $m=\frac{\theta n}{4}$
   
   \For{$k=0,\dots,\frac{4}{\theta}-8$}{

      Let $I=J=\{km,km+1,\dots,(k+8)m\}$. 
      
      \For{$i= \log 1/\theta,\dots,0$}{
         
          Set $\epsilon_i=2^{-i}$.
         
          Invoke \DSR{}$(G(I \times J),n,w_1,\frac{d}{\epsilon_i}$, $\frac{\epsilon_i}{8},\epsilon_i)$ to get $\mathcal{S}$ and $\mathcal{R}_1$.
         
\If{$\mathcal{S} \neq \emptyset$}{

          Invoke \SSES{}$(G(I\times J),\mathcal{S},$ $n,w_1,w_2,\frac{d}{\epsilon_i}$, $\frac{\epsilon_i}{8},\epsilon_i, \theta)$ to get $\mathcal{R}_2$.
}
\Else{

	 $\mathcal{R}_2=\emptyset$.

}
         
          Add items from $\mathcal{R}_1$ to $\mathcal{R}_D$ and from $\mathcal{R}_2$ to $\mathcal{R}_E$.
         
      }
      
   }
   
    Output $\mathcal{R}=\mathcal{R}_D \cup \mathcal{R}_E$.
   \caption{\CA{}$(x,y,n,w_1,w_2,d,\theta)$ }
    \label{alg-pI}
\end{algorithm}

\begin{algorithm}

   \Input{$G=G_{x,y}(I_G \times J_G)$ for some $I_G,J_G\subseteq \{0,1,\dots,n\}$, $w,d\in [n]$, the
endpoints of $I_G$ and $J_G$ are multiples of $w$ and $\delta,\epsilon \in [0,1]$.}
   
   \Output{Set $\mathcal{S}$ which is a subset of the $w$-decomposition of $I_G$ and a set $\mathcal{R}$ of $\delta$-aligned certified $w$-boxes
all with distance bound $5\epsilon_i$.}

   \vspace{1mm}
   \hrule\vspace{1mm}
    \
    Initialization: $\mathcal{S}=\mathcal{R}=\emptyset$. $\mathcal{T}=\mathcal{I}_{w}(I_G)$.
   
    $\mathcal{B}$, the set of $y$-candidates, is the set of width $w$ $\delta$-aligned subintervals of $J_G$ (having endpoints a multiple of $\delta w$.)

   \While{$\mathcal{T}$ is non-empty}{
       Pick $I\in \mathcal{T}$
      
       Sample $c_0 |\mathcal{B}| \frac{1}{d} \log n$ intervals $J\in \mathcal{B}$ uniformly at random and for each test  if $\ed(x_I,y_J) \leq \epsilon$.
      
      
      \If{for at most $\frac{c_0}{2} \log n$ sampled $J$'s, $\TE(x_I,y_J,\epsilon)<\infty$ }{

           $\mathcal{S} = \mathcal{S} \cup \{I\}$; $\mathcal{T}=\mathcal{T}- \{I\}$. ({\em $I$ is declared sparse}) 
         }
      \Else{
 
           ({\em $I$ is declared dense and used as a pivot})

          Compute: 
         
          $\mathcal{Y} = \{J \in \mathcal{B};\; \TE(x_I,y_J,3\epsilon) < \infty\}$.
         
          $\mathcal{X} = \{I' \in \mathcal{T};\; \TE(x_I,x_{I'},2\epsilon) < \infty\}$.
         
          Add $(I'\times J',5\epsilon)$ to $\mathcal{R}$ for all pairs $(I',J')\in \mathcal{X} \times \mathcal{Y}$.


           $\mathcal{T} = \mathcal{T} - \mathcal{X}$.
         
      }
      
   }
   
    Output $\mathcal{S}$ and $\mathcal{R}$.

   \caption{\DSR{}$(G,n,w,d,\delta,\epsilon)$}
    \label{alg-dsr}
\end{algorithm}

\begin{algorithm}

   \Input{$G=G_{x,y}(I_G,J_G)$ with $I_G,J_G\subseteq \{0,1,\dots,n\}$, $w_1,w_2,d,n$ are powers of 2, with $w_1,w_2,d <n$ and $w_1<w_2$. 
Endpoints of $I_G$ and $J_G$ are multiples of $w_2$,
$\mathcal{S}$ is a subset of the $w_1$-decomposition of $I_G$ and $\delta,\epsilon,\theta$ are non-positive integral powers of 2.}
   
  \Output{A set $\mathcal{R}$ of certified $w_2$-boxes in $G$.}
   
   \vspace{1mm}
   \hrule\vspace{1mm}
   
    Initialization: $\mathcal{R}=\emptyset$.
   
   
  $\mathcal{B}$, the set of $y$-candidates, is the set of width $w$ $\delta$-aligned subintervals of $J_G$ (endpoints are multiples of $\delta w$.)

\For{$I' \in \mathcal{I}_{w_2}(I_G)$}{
			\If{$\mathcal{S}$ includes a subset of $I'$}{

       Select $c_1 \log^2 n$ intervals $I$ independently and uniformly at random from 
$\mathcal{I}_{w_1}(I')\cap \mathcal{S}$, to obtain $\mathcal{H}$. 
         
      \For{each  $I \in \mathcal{H}$ and each $J\in \mathcal{B}$}{
         
         \If{$\TE(x_I,y_J,\epsilon)<\infty$}{
             Let $J'$ be such that $I' \times J'$ is the adjusted diagonal extension of $I \times J$ in $I'\times J_G$.
            
            Let $p=\TE(x_{I'},y_{J'},3\epsilon)$
            
            \If{$p < \infty$}{
                For $k=0,\dots,\log n $, add $(I'\times J', p + \theta + 2^{-k})$ to $\mathcal{R}$.
               
            }
           

         }
      }
}
      
   }
   
    Output $\mathcal{R}$.
   \caption{\SSES{}$(G,\mathcal{S},n,w_1,$ $w_2,d,\delta,\epsilon,\theta)$
}
   \label{alg-sses}
\end{algorithm}

For the analysis we must 
prove that $\mathcal{R}$ contains an "adequate approximation" of some min-cost alignment path $\tau$. To state this precisely,
we start with definitions and observations that
formalize intuitive notions from the overview.

\noindent
{\bf Cost and normalized cost.}
The {\em cost} of a path $\tau$, $\cost(\tau)$,  
from $(u_1,u_2)$ to $(v_1,v_2)$ in a grid-graph (see Section~\ref{sec:intro}), is the sum of the edge costs,
and the {\em normalized cost} is $\ncost(\tau)=\frac{\cost(\tau)}{v_1-u_1}$.  $\cost(G(I \times J))$ (or simply $\cost(I\times J)$), the {\em cost of 
subgraph $G(I \times J)$}, is the min-cost of a path from the
lower left to the upper right corner. The {\em normalized cost} is $\ncost(I\times J)=\frac{1}{\mu(I)}\cdot \cost(I \times J)$.

We note the following simple fact without proof:
\begin{proposition} 
\label{prop:sym diff}
For $I,J,J' \subseteq \{0,\ldots,n\}$,
$|\cost(I \times J)-\cost(I \times {J'})| \leq |J \Delta J'|$, where $\Delta$ denotes symmetric difference.
\end{proposition}
\noindent
{\bf Projections and subpaths.} The {\em horizontal projection} of a path $\tau=(i_1,j_1),\ldots,(i_{\ell},j_{\ell})$ is the set of $\{i_1,\ldots,i_{\ell}\}$.
We say that {\em $\tau$ crosses box $I \times J$ } if 
the vertices of $\tau$ belong to $I \times J$ and its horizontal projection 
is $I$. If the horizontal projection of $\tau$ contains $I'$, $\tau_{I'}$
denotes the (unique) minimal subpath of $\tau$ whose projection is $I'$.

\begin{proposition}
\label{prop:avg-cost}
Let $\tau$ be a path with horizontal projection $I$, and let
$I_1,\ldots,I_\ell$ be a decomposition of $I$.
Then the $\tau_{I_j}$ are edge-disjoint and so:

\begin{eqnarray*}
\cost(\tau) &\geq &\sum_{i=1}^\ell \cost(\tau_{I_i})\\
\ncost(\tau) &\geq &\sum_{i=1}^\ell \frac{\mu(I_i)}{\mu(I)}\ncost(\tau_{I_i}).
\end{eqnarray*}
\end{proposition}

\begin{definition}{\bf $(1-\delta)$-cover.}
\label{def:cover}
Let $\tau$ be a path with horizontal projection $I$ and let $I' \times J'$ be a (not
necessarily square) box
with $I' \subseteq I$. 
For $\delta\ge 0$
the box $I' \times J'$ {\em $(1-\delta)$-covers  $\tau$}
if the initial, resp. final, vertex of the subpath $\tau_{I'}$
is within $\delta \mu(I')$ vertical units of $(\min(I'), \min(J'))$, resp.  $(\max(I'),\max(J'))$.
\end{definition}

We will be interested in $(1-\delta)$-cover when $\delta \le 1/2$ but for simplicity we allow $\delta > 1/2$.

\begin{proposition}
\label{prop:cover}
Let $I' \times J'$ be a (not necessarily square) box that  $(1-\delta)$-covers path $\tau$. 
\begin{enumerate}
\item $\ncost(I' \times J') \leq \ncost(\tau_{I'})+2\delta$.
\item If $J''$ is any vertical interval such that $J' \cap J'' \ne \emptyset$, then $I' \times J''$
$(1-\delta-|J' \Delta J''|/\mu(I'))$-covers $\tau$.
\end{enumerate}
\end{proposition}

\begin{proof}
For the first part, let $J^0$ be the vertical projection of $\tau_{I'}$.
Then $\ncost(I' \times J^0) \leq \ncost(\tau_{I'})$ since $\tau_{I'}$ joins the lower left corner
of $I' \times J^0$ to the upper right corner.  Since $I' \times J'$
$(1-\delta)$-covers $\tau$, $|J' \Delta J^0| \leq 2\delta \mu(I')$, 
and by Proposition~\ref{prop:sym diff}, $\ncost(I' \times J') \leq \ncost(\tau_{I'})+2\delta$.

For the second part, observe that the vertical distance between the lower (resp. upper) corners
of $I' \times J'$ and $I' \times J''$ is at most $|J' \Delta J''|$.
\end{proof}

\noindent{\bf $\delta$-aligned boxes.}
A $y$-interval $J$ of width $w$ 
is {\em $\delta$-aligned} for $\delta \in (0,1]$ if its endpoints are multiples of $\delta w$ (which we require to be an integer).

\begin{proposition}\label{prop:align-box}
Let $\tau$ be a path that crosses  $I \times J$.
Suppose that  $I' \subseteq I$ has width $w$, and $\mu(J) \geq w$.
\begin{enumerate}
\item There is an interval $J^1$ with $\mu(J^1)=\mu(I')$ so that
$\ncost(I'\times J^1) \leq 2\ncost(\tau_{I'})$ and $I' \times J^1$ $(1-\ncost(\tau_{I'}))$-covers $\tau$.
\item
There is a $\delta$-aligned interval $J'$ of width $w$ so that 
$\ncost(I' \times J') \leq 2\ncost(\tau_{I'})+\delta$
and $I' \times J'$  $(1-\ncost(\tau_{I'})-\delta)$-covers $\tau.$
Moreover, if $J$ is $\delta$-aligned and $w|\mu(J)$ then $J' \subseteq J$.
\end{enumerate}
($J^1$, $J'$ are ``$\tau$-matches'' for $I'$, in the sense of the overview.)
\end{proposition}

\begin{proof}
Let $\tau'=\tau_{I'}$ be the min-cost subpath of $\tau$ that projects to $I'$.
Let $J^0$ be the vertical projection of $\tau'$. Note that $|\mu(J^0)-\mu(I')| \leq \cost(\tau')$.
Arbitrarily choose an interval $J^1$ of width
$\mu(I')$ that either contains or is contained in $J^0$.  Then 
$|J^0 \Delta J^1|=|\mu(J^0)-\mu(I')| \leq \cost(\tau')$, so  by Proposition
~\ref{prop:sym diff} $\ncost(I' \times J^1) \leq 2\ncost(\tau')$.  Furthermore $I' \times J^1$ $(1-\ncost(\tau'))$-covers $\tau'$.  
Let $J'$ be the closest $\delta$-aligned interval to $J^1$, so
$|J' \Delta J^1| \leq \delta \mu(I')$ and so $\ncost(I' \times J') \leq \ncost(I' \times J^1) + \delta \leq 
2 \ncost(\tau')+\delta$.  Finally since $I' \times J'$ 
is a vertical shift of $I' \times J^1$
of  normalized length at most $\delta$, we have $I' \times J'$ $(1-\ncost(\tau')-\delta)$-covers $\tau'$. 
\end{proof}

\begin{definition}
\label{def:DE}
\begin{enumerate}
\item The {\em main diagonal} of a box is the segment joining the lower
left and upper right corners.  
\item For a square box $I' \times J'$, and $I' \subseteq I$, 
the {\em true diagonal extension} of $I' \times J'$ to $I$ is the square
box $I \times \hat{J}$ 
whose main diagonal contains the main diagonal of $I' \times J'$ (in the infinite grid graph on $\Z \times \Z$).
\item
For a $w$-box $I' \times J'$ contained in strip $I \times J$, 
the {\em adjusted diagonal extension of $I' \times J'$ within $I \times J$} 
is the box $I \times J''$ obtained from the true diagonal extension of $I' \times J'$ to $I$
by the minimal vertical shift so that it is a subset of $I \times J$.
(The adjusted diagonal extension is the true diagonal extension if
the true diagonal extension is contained in $I \times J$; otherwise it's lower edge
is $\min(J)$ or its upper edge is $\max(J)$.)
\end{enumerate}
\end{definition}

\begin{center}
	\begin{figure}[ht]
\centering
\begin{tikzpicture}[scale=0.3,shorten >=1mm,>=latex]
 \tikzstyle gridlines=[color=black!20,very thin]
 


\draw[color=white,thin,fill=gray!30] (-20,18) rectangle (-6,30);

\draw[-,dashed] (-20,10) -- (-6,10);
\draw[-,dashed] (-20,28) -- (-6,28);
\draw[-,dashed] (-20,6) -- (-20,32);
\draw[-,dashed] (-6,6) -- (-6,32);
\draw[color=black,thick,fill=white] (-15,22.25) rectangle (-13,24);
\draw[-,dotted,thick] (-20,18) -- (-6,30);


\node at (-14,21.5)[scale=.8] {$I'$};
\node at (-12.25,23.25)[scale=.8] {$J'$};
\node at (-21,19) {$J$};
\node at (-13,8.5) {$I$};
\node at (4.5,23.5) {$\hat{J}$};
\node at (-4.5,23.5) {$\hat{J}$};

\draw [black,<->,thin] (-20.5,10) -- (-20.5,28.25);
\draw [black,<->,thin] (-5.5,18) -- (-5.5,30);
\draw [black,<->,thin] (-20,9.5) -- (-6,9.5);




\draw[color=white,thin,fill=gray!30] (0,16) rectangle (14,28);

\draw[-,dashed] (0,10) -- (14,10);
\draw[-,dashed] (0,28) -- (14,28);
\draw[-,dashed] (0,6) -- (0,32);
\draw[-,dashed] (14,6) -- (14,32);
\draw[color=black,thick,fill=white] (5,22.25) rectangle (7,24);
\draw[-,dotted,thick] (0,18) -- (14,30);


\node at (6,21.5)[scale=.8] {$I'$};
\node at (7.75,23.25)[scale=.8] {$J'$};
\node at (-1,19) {$J$};
\node at (7,8.5) {$I$};
\node at (15.5,21.5) {${J''}$};

\draw [black,<->,thin] (-.5,10) -- (-.5,28.25);
\draw [black,<->,thin] (14.5,16) -- (14.5,28);
\draw [black,<->,thin] (0,9.5) -- (14,9.5);

 \node at (-13,5)[scale=1] {(a)};
\node at (7,5)[scale=1] {(b)};

\end{tikzpicture}
\caption{Illustration of diagonal extension: Given a square box $I'\times J'$, the true diagonal extension is the grey box $I\times \hat{J}$ in $(a)$ and the adjusted diagonal extension is the grey box $I\times {J''}$ in $(b)$.}
   \label{fig:diagonal-extension}
\end{figure}

\end{center}

\begin{proposition}\label{prop:extension}
Suppose path $\tau$ crosses $I \times J$ and $\ncost(\tau_I) \leq \epsilon$. Let $w=\mu(I)$. Let $I' \times J'$ be a $w'$-box  that $(1-\delta)$-covers 
$\tau_{I'}$ and $I'\subseteq I$.  
Then the adjusted diagonal extension $I \times J''$ of $I' \times J'$  within $I \times J$  $(1-(\epsilon + \delta  \frac{w'}{w}))$-covers $\tau$ and  satisfies $\ncost(I \times J'') \leq 3\epsilon + 2\delta \frac{w'}{w}$. 
	\end{proposition}
	
\begin{proof}
It suffices to show that $I \times J''$ $(1-(\epsilon+\delta w'/w))$-covers
$\tau$, since then Proposition~\ref{prop:cover} gives us the needed upper bound
on $\ncost(I \times J'')$.

\noindent
{\bf Case 1.} $I \times J''$ is equal to the
true diagonal extension.
Let ${\tau}_{I}$, ${\tau}_{I'}$ be the min-cost subpath of $\tau$ that projects on $I$ and $I'$ respectively.   

We will give an upper bound  on the vertical distance from the final vertex of $\tau_I$
to the upper right corner of $I \times J''$. Let ${\tau}_{u}$ be the
subpath of $\tau$ that starts at the final vertex of ${\tau}_{I'}$ and ends at the final vertex of ${\tau}_{I}$.  Let $I_u$ and $J_u$ be the horizontal and vertical projections of ${\tau}_{u}$.
The start vertex of $\tau_u$ has vertical distance at most $\delta w'$ from the
main diagonal of $I \times J''$.  The final vertex of $\tau_u$ therefore
has vertical distance at most $\delta w'+ |\mu(I_u)-\mu(J_u)|$ from the
upper corner of $I \times J''$, and this is at most $\delta w' + \epsilon w$, since $\cost(\tau) \geq |\mu(I_u)-\mu(J_u)|$.  A similar argument gives the same upper bound on the vertical distance between the start vertex of $\tau_I$ and the lower left corner of $I \times J''$, so $G''(I \times J'')$  $(1-(\epsilon + \delta  w'/w))$-covers $\tau$.

\noindent
{\bf Case 2.}
$I \times J''$ is not the
true diagonal extension.   Extend the set $J$ to $\bar{J}$ by adding $\mu(I)$
elements before and after.  (It is possible that $\bar{J}$ is not a subset
of $\{0,\ldots,n\}$; in this case we imagine that $y$ is extended to a sequence
$y^*$ by adding $\mu(I)$ new symbols to the beginning and end of $y$ and that
we are in the grid graph $G_{x,y^*}$.) Let $I \times J'''$ be the adjusted diagonal extension of $I' \times J'$
to $I \times \bar{J}$.  This is equal to the true diagonal extension, and so by Case 1, 
$I \times J'''$ $(1-(\epsilon + \delta  w'/w))$-covers $\tau$.   We claim that
$I \times J''$ does also.  Assume $J'''$ falls
below $\min(J)$ (the case that $J'''$ is above $\max(J)$ is similar). 
Then $I \times J''$ is obtained by shifting $I \times J'''$
up until the lower edge coincides with $\min(J)$.
The lower vertex of $\tau_I$ has $y$-coordinate at least $\min(J)$.  

If the $y$-coordinate of the upper vertex of $\tau_I$ is at most $\max(J'')$, then
$J''$ contains vertical projection of $\tau_I$, and $I \times J''$
$(1-\epsilon)$-covers $\tau$.
If the $y$-coordinate of the upper vertex of $\tau_I$ is greater than  $\max(J'')$,
shifting $I \times J'''$ up to $I \times J''$ can only
decrease the vertical distance from the the lower left corner to the start of $\tau_I$
and from the upper corner to the end of $\tau_I$, so $I \times J''$ $(1-(\epsilon + \delta  w'/w))$-covers $\tau$. 
\end{proof}

	

\smallskip\noindent
{\bf $(k,\zeta)$-approximation of a path. }
This formalizes the
notion of adequate approximation of a path  by a certified box sequence.
\begin{definition}
\label{def:approx}
Let $G$ be the grid graph on $I\times J$. Let $\zeta \in [0,1]$. Let $\tau$ be a path that crosses $G$. A sequence of certified boxes $\sigma=\{(I_1\times J_1,\epsilon_1),(I_2 \times J_2,\epsilon_2),\dots, (I_\ell \times J_\ell,\epsilon_\ell)\}$  
 {\em $(k,\zeta)$-approximates} $\tau$ provided that:
\begin{enumerate}
\item $I_1,\ldots,I_{\ell}$ is a decomposition of $I$. 
\item For each $i\in [\ell]$, $I_i \times J_i$ $(1-\epsilon_i)$-covers $\tau$.
\item $\sum_{i\in [\ell]}  \epsilon_i \mu(I_i) \le (k \cdot \ncost(\tau) + \zeta) \mu(I)$.
\end{enumerate}
\end{definition}
 
\begin{proposition}
\label{prop:approx decomp}
Suppose path $\tau$ crosses  $I \times J$ and $I_1,\ldots, I_m$ is a decomposition of $I$, and 
for $i\in [m]$, $\sigma_i$ is a certified box sequence that $(k,\zeta)$-approximates
$\tau_{I_i}$. Then $\sigma_1,\ldots,\sigma_m$ $(k,\zeta)$-approximates
$\tau$.
\end{proposition}

\begin{proof}
It is obvious that $\sigma_1,\ldots,\sigma_m$ is a sequence of certified boxes, that the horizontal projections of all the boxes form a
decomposition of $I$ and that each box $(I'_i,J'_i,\epsilon_i)$ $(1-\epsilon_i)$-covers $\tau$.
The final condition is verified by splitting the sum on the left into $m$ sums where
the $j$th sum includes terms for $I'_i \subseteq I_j$, and is bounded above by
 $(k \cdot \ncost(\tau_{I_j})+\zeta)\mu(I_j)$. Summing the latter sum over $j$ and using Proposition
~\ref{prop:avg-cost} we get that $\sigma_1,\ldots,\sigma_m$ $(k,\zeta)$-approximates the path $\tau$.
\end{proof}

\smallskip\noindent
{\bf $(d,\delta,\epsilon)$-dense and -sparse.} 
Fix a box $I \times J$. An interval $I' \subseteq I$ of width $w$ is {\em $(d,\delta,\epsilon)$-sparse} (wrt $I \times J$) for integer $d$ and $\epsilon,\delta \in (0,1]$
if there are at most $d$ $\delta$-aligned $w$-boxes in $I' \times J$ of $\ncost$ at most $\epsilon$, and is 
{\em $(d,\delta,\epsilon)$-dense} otherwise.

\smallskip\noindent
{\bf  The sets $\mathcal{S}_i$ and $\mathcal{S}_i(I')$.}  For fixed $k$ in the outer loop of \CA{},
the set $\mathcal{S}$ created in iteration $i$ of \CA{} is denoted by $\mathcal{S}_i$. For any interval $I'$, $\mathcal{S}_i(I')$ is the set
of subintervals of $I'$ belonging to $\mathcal{S}_i$.

\smallskip\noindent 
{\bf Successful Sampling.}
The algorithm uses random sampling in two places: in \DSR{} to test density of width $w_1$ intervals $I$, and in \SSES{} to sample sparse intervals $I \in \mathcal{I}_{w_1}(I')$.

We now specify what we need from the random sampling. 

\begin{definition}
\label{def:sampling} We define two events where the latter depends on the former:
\begin{itemize}
\item A run of the algorithm has {\em successful classification sampling} provided that
for all $k \in \{0,\ldots,$ $4/\theta-8\}$ and $i \in \{0,\ldots,\log \frac{1}{\theta}\}$ in the nested \CA{} loops, in each call to \DSR(),
for every width $w_1$ interval $I$ with endpoints a multiple of $w_1$,
if $I$ is $(\frac{d}{\epsilon_i},\frac{\epsilon_i}{8},\epsilon_i)$-dense interval (in terms of global \CA{}'s parameters with respect to \CA{}'s $I \times J$), \DSR{} does not
assign $I$ to $\mathcal{S}$ and
if $I$ is $(\frac{d}{4\epsilon_i},\frac{\epsilon_i}{8},\epsilon_i)$-sparse and \DSR() picks it from $\mathcal{T}$, \DSR{} places $I$  in $\mathcal{S}$.

\item Let $\tau$ be a source-sink path in $G_{x,y}$. 
For each $I' \in \mathcal{I}_{w_2}$, for $I=J=\{km,km+1,\dots,(k+8)m\}$, where $k \in \{0,\dots,\frac{4}{\theta}-8\}$ 
is the largest integer such that $km \le \min (I') - \frac{\theta n}{2}$ or 0 if $\min (I') < \frac{\theta n}{2}$, define set $\mathcal{S}_i(I')$ to be $\mathcal{S} \cap \mathcal{I}_{w_1}(I')$ where $\mathcal{S}$ is returned by the invocation of \DSR{}$(G(I \times J),n,w_1,\frac{d}{\epsilon_i}$, $\frac{\epsilon_i}{8},\epsilon_i)$. Furthermore, define the set of {\em winners}
$\mathcal{W}_i(I') = \{ I'' \in \mathcal{S}_i(I'),\, \epsilon_i \ge 3\ncost(\tau_{I''}) + \ncost(\tau_{I'}) + \theta\}$.

We say a run of the algorithm has {\em successful extension sampling for $\tau$} provided that: for each $i \in \{0,\ldots,\log \frac{1}{\theta}\}$ and $I' \in \mathcal{I}_{w_2}$, if $|\mathcal{W}_i(I')| \ge \frac{1}{64} |\mathcal{S}_i(I')|$ and both are non-empty then at least one interval from $\mathcal{W}_i(I')$ will be included in the set $\mathcal{H}$ sampled in iteration $I'$ of the outer loop during invocation of \SSES{}$(G(I\times J),\mathcal{S},$ $n,w_1,w_2,\frac{d}{\epsilon_i}$, $\frac{\epsilon_i}{8},\epsilon_i, \theta)$.
\end{itemize}
We say a run of the algorithm has {\em successful sampling for $\tau$} if it has successful classification sampling and successful extension sampling for $\tau$.
\end{definition}

We will need the following variant of  the Chernoff bound.

\begin{proposition}[Chernoff bound]\label{lem:chernoff}
There is a constant $c_0$ such that the following is true. Let $1\le d \le n$ be integers, $B$ be a set and $E\subseteq B$.
Let us sample $c_0 \frac{|B|}{d} \log n$ samples from $B$ independently at random with replacement.
\begin{enumerate}
\item If $|E| \ge d$ then the probability that less than $\frac{c_0}{2} \log n$ samples are from $E$ is at most $1/n^{10}$.
\item If $|E| \le d/4$ then the probability that at least $\frac{c_0}{2} \log n$ samples are from $E$ is at most $1/n^{10}$.
\end{enumerate}
\end{proposition}

\begin{proposition}
\label{prop:sampling}
For large enough $n$ and any source-sink path in $G_{x,y}$,
a run of  \CA{} has {\em successful sampling for $\tau$} with probability at least $1-n^{-7}$.
\end{proposition}

\begin{proof}
 By Proposition \ref{lem:chernoff}, the probability that the condition of successful classification sampling fails for a particular $k,i,I$
is at most $n^{-10}$. The number of choices for $k,i,I$ is at most $\frac{4}{\theta} \cdot (1 + \log \frac{1}{\theta}) \frac{n}{w_1}
\leq n^2$ (for large enough $n$)  so the overall probability that  (1) fails 
is at most  $n^{-8}$.
 
The probability that the condition of successful extension sampling for $\tau$ fails for a particular $i,I'$ is $(1-\frac{1}{64})^{c_1 \log^2 n} \leq n^{-10}$.
The number of $i,I'$ is less than $n^2$ (for large enough $n$), so the overall failure probability is at most $n^{-8}$ .
\end{proof}

We assume that coins are fixed in a way that gives successful sampling for some shortest source-sink path in $G_{x,y}$.

\subsection{Properties of the covering algorithm}

The main property of \CA{} to be proved is:

\begin{theorem}\label{thm:pla}
Let $x,y$ be strings of length $n$, $1/n \le \theta \le 1 $ be a real. Let $w_1,w_2,d$ satisfy
$w_1\le \theta w_2$, $w_2 \le  \frac{\theta n}{4}$ and $1 \leq d \leq \frac{\theta n}{w_1}$. Assume $n,w_1,w_2,d,\theta$ are  powers of 2.
Let $\mathcal{R}$ be the set of weighted boxes obtained by running \CA{}$(x,y,n,w_1,w_2,d,\theta)$ with $c_1 > 640$.
Then 
\begin{enumerate}
\item Every $(I\times J,\epsilon)\in \mathcal{R}$ is correctly certified, i.e., $\ed(x_{I},y_{J})\le \epsilon$, and 
\item For every source-sink path $\tau$ in $G=G_{x,y}$ of normalized cost at most $\theta$, in a run that satisfies successful sampling for $\tau$, there is a subset of $\mathcal{R}$
that $(45,15\theta)$-approximates $\tau$. 
\end{enumerate}
\end{theorem}

\begin{proof}
All boxes output are correctly certified:
Each box in $\mathcal{R}_E$ comes from \SSES{} which
only certifies boxes with at least their exact edit distance.
For $(I\times J,\epsilon)\in \mathcal{R}_D$, there must be an $I'$ such that
$\ed(x_{I'},y_{J})\le \frac{3}{5} \cdot \epsilon$ and $\ed(x_{I'},x_{I})\le \frac{2}{5} \cdot \epsilon$ and so by triangle inequality $\ed(x_{I},y_{J})\le \epsilon$.

It remains to establish (2). Fix a source-sink path $\tau$ of
normalized cost $\le \theta$.
By Proposition~\ref{prop:approx decomp} it is enough to show that for each $I' \in \mathcal{I}_{w_2}$, $\mathcal{R}$ contains a box sequence that $(45,15\theta)$-approximates $\tau_{I'}$.  So we
fix $I' \in \mathcal{I}_{w_2}$.

The main loop (on $k$) of \CA{} processes $G$ in overlapping boxes.  
Since $\ncost(\tau) \leq \theta$, one of these boxes, which we'll call
$I \times J$, must contain $\tau_{I'}$: 

\begin{claim}\label{prop:pIa-diag-cover}
Let $I'\in \mathcal{I}_{w_2}$. There exist intervals $I,J \subseteq \N$, $I=J$
that are enumerated in the main loop of \CA{} such that $I' \subseteq I$ and $\tau_{I'}$ crosses $G(I' \times J)$.
\end{claim}

\begin{proof}
Since $\tau$ is of cost at most $\theta$, it cannot use more than $\theta n/2$ horizontal edges as for each horizontal edge
of cost 1, it must use one vertical edge of cost 1. Similarly for vertical edges. So $\tau$ is confined to diagonals $\{-\theta n/2,\dots,0,\dots,\theta n/2\}$
of $G$. By the choice of $m$ in \CA{}, there will be $I$ and $J$ considered in the main loop of the algorithm
such that $I' \subseteq I$ and $\tau_{I'}$ crosses $G(I \times J)$. In particular, $I=J=\{km,km+1,\dots,(k+8)m\}$, where 
$k \in \{0,\dots,\frac{4}{\theta}-8\}$ is the largest integer such that $km \le \min (I') - \frac{\theta n}{2}$ or 0 if $\min (I') < \frac{\theta n}{2}$,
has the desired property. 
\end{proof}

Let $I,J$ be as provided by the claim. Let $\mathcal{I}'$ be the $w_1$-decomposition of $I'$.  We will show one of the following must hold:
(1) $\mathcal{R}_D$ contains a sequence of certified $w_1$-boxes that $(45,15\theta)$-approximates $\tau_{I'}$, or (2)
There is a single certified $w_2$-box in $\mathcal{R}_E$ that $(45,15\theta)$-approximates $\tau_{I'}$.

Let $t=\log\frac{1}{\theta}$. For $i=t,\dots,0$, let $\epsilon_i = 2^{-i}$ and let
$\mathcal{S}_i$ be the set $\mathcal{S}$ obtained at the iteration $i$
of  \CA{}$(x,y,n,w_1,w_2,d,\theta)$.

We note:

\begin{claim}\label{claim:pIa-dense-box}
Let $i \in \{0,\dots,\log 1/\theta\}$. Suppose $I'' \in \mathcal{I}_{w_1}(I)$ and $J'' \subseteq J$ is $\epsilon_i/8$-aligned.
If $I'' \not\in \mathcal{S}_i$ and $\ncost(I'' \times J'') \leq \epsilon_i$ then $(I'' \times J'',5 \epsilon_i) \in \mathcal{R}_D$.
\end{claim}

\begin{proof}
If $I'' \not\in \mathcal{S}_i$ then in the call to
\DSR{}$(G(I \times J),n,w_1,d/\epsilon_i,\epsilon_i/8,\epsilon_i)$ 
there is an iteration of the main
loop, where  the selected interval $\tilde{I}$ from $\mathcal{T}$ is declared
dense and $\ed(x_{\tilde{I}},x_{I''}) \le 2 \epsilon_i$.
Since $\ed(x_{I''},y_{J''}) \le \epsilon_i$, $\ed(x_{\tilde{I}},y_{J''}) \le 3 \epsilon_i$ and so
$I'' \in \mathcal{X}$ and $J'' \in \mathcal{Y}$.
Thus, \DSR{} certifies $(I'' \times J'',5 \epsilon_i)$, which is added to $\mathcal{R}_D$.
\end{proof}

The theorem follows from:
\begin{claim}\label{claim:main}
For the interval $I'\in \mathcal{I}_{w_2}$, assuming successful sampling for $\tau$,
$\mathcal{R}_E$ or $\mathcal{R}_D$ contains a $(45,15\theta)$-approximation of $\tau_{I'}$.
\end{claim}
The proof is similar to that of Claim~\ref{claim:main overview},
with adjustments  for some technicalities. 

\begin{proof}
Let $\tau'=\tau_{I'}$ and $\kappa=\ncost(\tau')$.
Let $\mathcal{I}'=\mathcal{I}_{w_1}(I')$. For $I'' \in \mathcal{I}'$, let $\kappa_{I''}=\ncost(\tau_{I''})$. By Proposition \ref{prop:align-box}, for all $I''\in \mathcal{I}'$ and $\epsilon_i \ge \kappa_{I''}$ there is an $\epsilon_i/8$-aligned vertical interval $J_i^{\tau}(I'') \subseteq J$, such that $\ncost(I''\times J_i^{\tau}(I'')) \leq 2\kappa_{I''}+\epsilon_i/8$ and $I'' \times J_i^{\tau}(I'')$  $(1-\kappa_{I''}-\epsilon_i/8)$-covers $\tau_{I'}$.

Let $s(I'')$ be the largest integer such that 
$\epsilon_{s(I'')} \ge 3 \kappa_{I''} + \kappa + \theta $. Let $t(I'') \le s(I'')$ be the largest integer such that 
$I'' \not\in \mathcal{S}_{t(I'')}$. (Since $\theta n/w_1 \ge d$, $\mathcal{S}_0 = \emptyset$, so $t(I'')$ is well-defined.) Let $a(I'')=\epsilon_{s(I'')}$ (this
plays a similar role to $u(I'')$ in Section~\ref{sec:overview}) and $b(I'')=\epsilon_{t(I'')}$.

For all $\epsilon_i \in [a(I''),b(I'')]$,  $\ncost(I''\times J_i^{\tau}(I'')) \leq \epsilon_i$ and $I'' \times J_i^{\tau}(I'')$ $(1-\epsilon_i)$-covers $\tau'$. By the definition of $b(I'')$ and Claim \ref{claim:pIa-dense-box}, $\mathcal{R}_D$ contains the certified box $(I'' \times J_{t(I'')}^{\tau}(I''),5b_{I''})$. So $\mathcal{R}_D$ contains a $(45,15\theta)$-approximation of $\tau'$ if it happens that:
\begin{equation}
\label{cond dense cover}
\sum_{I''\in \mathcal{I}'}5 b(I'') \le \frac{45}{8} \sum_{I''\in \mathcal{I}'}a(I'').
\end{equation}
Indeed, since $a(I'')=\epsilon_{s(I'')} \le 2(3\kappa_{I''}+\kappa+\theta)$, $\sum_{I''\in \mathcal{I}'} \mu(I'') = \mu(I')$, $\sum_{I''\in \mathcal{I}'} \kappa_{I''} \mu(I'') \le \kappa \mu(I')$ (by Proposition~\ref{prop:avg-cost}), under Equation~\ref{cond dense cover} being true:
\begin{equation*}
\sum_{I''\in \mathcal{I}'}5 b(I'') \mu(I'') \le \frac{45}{8} \sum_{I''\in \mathcal{I}'} (6\kappa_{I''}+2\kappa+2\theta) \mu(I'') \le (45 \kappa + 15\theta) \mu(I').
\end{equation*}

Next we determine a sufficient condition that $\mathcal{R}_E$ contains a box sequence (consisting of a single box)
 that $(5,4\theta)$-approximates $\tau'$.  We will show that if this condition does not hold then Equation~\ref{cond dense cover} is true.
Let $\mathcal{S}_i(I')=\mathcal{S}_i \cap \mathcal{I}'$. Interval $I'' \in \mathcal{S}_i(I')$  is a {\em winner for iteration $i$} if $\epsilon_i \ge a(I'')$.  
This set of winners is denoted by $\mathcal{W}_i(I')$.  It suffices that during iteration $i$, the
set of $c_1\log^2n$ samples taken in 
\SSES{} includes a winner $I''$; then since
$\ed(I'',J_i^{\tau}(I'')) \leq \epsilon_i$,  the (adjusted) diagonal extension $I' \times \tilde{J}$ of $I'' \times J_i^{\tau}(I'')$ will be certified.
By Proposition \ref{prop:extension}, $I' \times \tilde{J}$ has normalized cost at most $3\kappa + 2 \epsilon_i w_1 / w_2 \le 3 \kappa + 2\theta \le 3 \epsilon_i $
and it $(1-(\kappa + \theta))$-covers $\tau'$. If $\kappa=0$ then  $(I' \times \tilde{J}, \ncost(I'\times \tilde{J}) + \theta + 2^{- \log n } )$ is in  $\mathcal{R}_E$ by the behavior of \SSES{} and it $(5,4\theta)$-approximates $\tau'$. 
Otherwise $\kappa \ge 1/n$; so set $k=\lfloor \log 1/\kappa \rfloor$. Thus, $k \le \log n$ and $2^{-k} \in [\kappa,2\kappa)$. Then $(I'\times \tilde{J}, \ncost(I'\times \tilde{J}) + \theta + 2^{-k} )$ is in $\mathcal{R}_E$ and it $(5,4\theta)$-approximates $\tau'$.

Under successful sampling of extensions for $\tau$ if
$|\mathcal{W}_i(I')| \ge \frac{1}{64} |\mathcal{S}_i(I')|$, at least one interval from $\mathcal{W}_i(I')$ will be included in our $c_1 \log^2 n$ samples during \SSES{} and 
$\mathcal{R}_E$ will contain a $(5,4\theta)$-approximation of $\tau'$ as above. 
So suppose this condition is not met, i.e., for all $i$, $|\mathcal{W}_i(I')| < \frac{1}{64} |\mathcal{S}_i(I')|$. This implies:

\begin{equation}
\label{eqtn dense}
\text{For all $i$, } |\mathcal{W}_i(I')| < \frac{1}{32} |\mathcal{S}_i(I')-\mathcal{W}_i(I')|.
\end{equation}
 We show that  this in turn implies (\ref{cond dense cover}). Multiplying (\ref{eqtn dense}) by $\epsilon_i$ 
and summing on $i$ yields:
 \begin{equation}
 \label{eqtn winner loser}
 \sum_{I'' \in \mathcal{I}'} \sum_{i:I''\in \mathcal{W}_i(I')} \epsilon_i < \frac{1}{32} \sum_{I'' \in \mathcal{I}'} \sum_{i:I''\in \mathcal{S}_i(I') - \mathcal{W}_i(I')} \epsilon_i .
 \end{equation}
$I'' \in \mathcal{S}_i(I')-\mathcal{W}_i(I')$ implies $\epsilon_i < a(I'')$. Summing the geometric series: 
\begin{equation}
\label{eqtn series sum}
\sum_{i:I''\in \mathcal{S}_i(I') - \mathcal{W}_i(I')} \epsilon_i \le 2 a(I'').
\end{equation}

Either $a(I'')=b(I'')$ or $a(I'') < b(I'')$. If the latter, then 
$I'' \in \mathcal{W}_i(I')$ for $\epsilon_i = b(I'')/2$. So:
\begin{align*}
\sum_{I'' \in \mathcal{I}'}b(I'') &\le \sum_{I''} \Big( a(I'') + \sum_{i : I'' \in \mathcal{W}_i(I')}2 \epsilon_i \Big)\\
& < \sum_{I''} \Big( a(I'') +  \frac{1}{16}  \sum_{i : I'' \in \mathcal{S}_i(I') - \mathcal{W}_i(I')} \epsilon_i \Big)\\
& \le \frac{9}{8} \sum_{I''\in \mathcal{I}'}a(I'')
\end{align*}
which implies Equation~\ref{cond dense cover}. (The second inequality follows from (\ref{eqtn winner loser}) and the last inequality from (\ref{eqtn series sum}).)
\end{proof}

\end{proof}

\subsection{Time complexity of \CA{}}
We write $t(w,\epsilon)$ for the time of $\TE(z_1,z_2,\epsilon)$ on
strings of length $w$. We assume $t(w,\epsilon) \ge w$, and that for $k \geq 1$, 
there is a constant $c(k)$ such that for all $\epsilon \in [0,1]$ and all $w>1$,
$t(w, k\epsilon) \le c(k) \cdot t(w,\epsilon) + c(k)$. As mentioned earlier, by~\cite{UKK85}, we can use $t(w,\epsilon)= O(w^2\epsilon)$. 
\begin{theorem}\label{thm:time-pIa}
Let $n$ be a sufficiently large power of 2 and $\theta \in [1/n,1]$
be a power of 2. Let $x,y$ be strings of length $n$.
Let $\log n\le w_1\le w_2 \le \theta n/4$, $1\le d\le n$ be powers of 2, where $w_1 | w_2$ and $w_2 |n$, and $w_1/w_2 \le \theta$. The size of the set $\mathcal{R}$ output by \CA{} is
$O((\frac{n}{w_1})^2 \log^2 n)$ and in 
any run that satisfies successful classification sampling, \CA{} runs
in time:
\begin{align*}
   O\Bigg( 
     &|\mathcal{R}|+\sum_{\substack{i= \log 1/\theta ,\dots,0 \\ \epsilon = 2^{-i}}} \Big(\frac{\theta n^2 \log n}{d \epsilon w_1^2} \cdot t(w_1,\epsilon) +
      \frac{\theta n^2 \log^2 n}{w_1 w_2 \epsilon}  \cdot t(w_1,\epsilon)
         + \frac{n d \log^2 n}{w_2 \epsilon} \cdot t(w_2,\epsilon)\Big)
     \Bigg).
\end{align*}
\end{theorem}

\begin{proof}
To bound $|\mathcal{R}|$ note that for each choice of $k,i$
in the outer and inner loops of \CA{}, the set of candidate boxes of width $w_1$  has size 
 $O(\frac{\theta n}{w_1}\frac{\theta n}{w_1\epsilon_i})$.  This upper bounds
 the number of boxes certified by \DSR{}.  The call to \SSES{} constructs at most
one diagonal extension for each such candidate box,
and each diagonal extension gives rise to at most $O(\log n)$ certified boxes.
Thus, for each $(k,i)$ there are $O(\frac{\theta^2 n^2 \log n}{(w_1)^2\epsilon_i})$ certified boxes.  Summing the geometric series over $i$, noting that $\min(\epsilon_i)=\theta$,
and summing over $O(1/\theta)$ values of $k$ gives the required bound on 
$|\mathcal{R}|$.

The steps in the algorithm that actually construct certified boxes cost $O(1)$ per box giving the first
term in the time bound.  

We next bound the other contributions to the running time.
The outer loop of  \CA{} has $\frac{4}{\theta}  -7$ iterations on $k$'s.
The inner loop has $1+\log \frac{1}{\theta}$ iterations on $i$. Each
iteration invokes \DSR{} and \SSES{} on $I \times J$ with $I$ and $J$ of width at most $4 \theta n$.  

We bound the time of a call to \DSR{}. To distinguish between
local variables of \DSR{}
and  global variables of \CA{}, we denote local input variables as $\hat{G},\hat{n},\hat{w},\hat{d},\hat{\delta},\hat{\epsilon}$.  
For $\mathcal{B}$ and $\mathcal{T}$ as in \DSR{}, $|\mathcal{T}| \le |\mathcal{B}| \leq \frac{\mu(I_{\hat{G}})}{\hat{\delta}\hat{w}}$, since $\mu(I_{\hat{G}})=\mu(J_{\hat{G}})$. The main while
loop of \DSR{} repeatedly picks intervals $I \in \mathcal{T}$ and samples $c_0|\mathcal{B}|\frac{\log \hat{n}}{\hat{d}}\leq \frac{c_0 \mu(I_{\hat{G}})\log \hat{n}}{\hat{d}\hat{\delta}\hat{w}}$ vertical intervals $J$ and tests whether $\ed(x_I,y_J) \leq \hat{\epsilon}$.
Each such test takes time $t(\hat{w},\hat{\epsilon})$. This is done at most once for
each of the $\mu(I_{\hat{G}})/\hat{w}$ horizontal candidates for a total time of
$O(\frac{\mu(I_{\hat{G}})^2 \log \hat{n}}{\hat{w}^2\hat{\delta}\hat{d}})t(\hat{w},\hat{\epsilon})$.  We next bound the cost of processing a pivot $I$.
This requires testing $\ed(x_I,y_J) \leq 3\hat{\epsilon}$ for $J \in \mathcal{B}$ and $\ed(x_I,x_{I'}) \leq 2\hat{\epsilon}$
for $I' \in \mathcal{T}$.  Each test costs $O(t(\hat{w},\hat{\epsilon}))$ (by our assumption
on $t(\cdot,\cdot)$), and since $|\mathcal{T}| \leq |\mathcal{B}|=\frac{\mu(I_{\hat{G}})}{\hat{w}\hat{\delta}}$,
$I$ is processed in time  $O(\frac{\mu(I_{\hat{G}})}{\hat{w}\hat{\delta}}t(\hat{w},\hat{\epsilon}))$.
This is multiplied by the number of intervals declared dense, which we now upper bound.  
If $I$ is declared dense then at the end of processing $I$, $\mathcal{X}$ is removed from $\mathcal{T}$.
This ensures $\ed(I,I') > 2\epsilon$ for any two intervals $I,I'$ declared dense. 
By the triangle inequality the sets 
$\mathcal{B}(I) = \{J \in \mathcal{B};\; \ed(x_{I},y_J) \le \epsilon\}$
are disjoint for different pivots.  By successful classification sampling, for each pivot $I$,
$|\mathcal{B}(I)| \geq \frac{\hat{d}}{4}$, and thus at most 
$|\mathcal{B}|/(\hat{d}/4) = \frac{4\mu(I_{\hat{G}})}{\hat{d}\hat{\delta}\hat{w}}$ intervals are
declared dense, so all intervals declared dense are processed in time
$O(\frac{\mu(I_{\hat{G}})^2}{\hat{w}^2 \hat{d}\hat{\delta}^2})t(\hat{w},\hat{\epsilon})$.

The time for dense/sparse classification of intervals
and for processing intervals declared dense is at most
$O(\frac{\mu(I_{\hat{G}})^2\log \hat{n}}{\hat{w}^2 \hat{d}\hat{\delta}^2})t(\hat{w},\hat{\epsilon})$.
During iteration
$i$ of the inner loop of \CA{}, the local variables of \DSR{}
are set as $\hat{n}=n$, $\mu(I_{\hat{G}})\leq 4\theta n$, $\hat{w}=w_1$, $\hat{d}=d/\epsilon_i$,
$\hat{\delta}=\epsilon_i/8$. Substituting these parameters yields
time $O(\frac{\theta^2n^2\log n}{(w_1)^2 d \epsilon_i})t(w_1,\epsilon_i)$.  Multiplying 
by the $O(1/\theta)$ iterations on $k$ gives the first summand of the
theorem.

Next we turn to \SSES{}.  The local input variables 
$n,w_1,w_2,\mathcal{S},\theta$ are
set to their global values so we denote them without $\hat{\quad}$.  The other
local input variables are denoted as $\hat{G},\hat{d},\hat{\delta}, \hat{\epsilon}$.
The local variable $\mathcal{B}$ has size $\frac{\mu(I_{\hat{G}})}{\hat{\delta}w_1}$. By successful classification sampling,
we assume that on every call,
 every interval in $\mathcal{S}$ is $(\hat{d},\hat{\delta},\hat{\epsilon})$-
sparse.
The outer loop enumerates the  $\mu(I_{\hat{G}})/w_2$ intervals $I'$
of  $\mathcal{I}_{w_2}(I_{\hat{G}})$. 
We select $\mathcal{H}$ to be $c_1 \log^2 n$ random subsets from subsets of  $I'$
belonging to $\mathcal{S}$.  For each $I \in \mathcal{H}$
and $J \in \mathcal{B}$, we call $\TE(x_I,y_J,\hat{\epsilon})$, taking   
time $t(w_1,\hat{\epsilon})$.   The total time of all tests is $O(\frac{\mu(I_{\hat{G}})^2\log^2 n}{\hat{\delta}w_1w_2})t(w_1,\hat{\epsilon})$.
Using  $\hat{d}=d/\epsilon_i$, $\hat{\delta}=\epsilon_i/8$
and $\hat{\epsilon}=\epsilon_i$ from the $i$th call to \SSES{}
gives $O(\frac{\theta^2n^2\log^2 n}{\epsilon_i w_1w_2})t(w_1,\epsilon_i)$. Multiplying 
by the $O(1/\theta)$ iterations on $k$ gives the second summand in  the theorem.

Assuming successful classification sampling, all intervals in the set $\mathcal{S}$ passed 
from \DSR{} to \SSES{} are
$(\hat{d},\hat{\delta},\hat{\epsilon})$-sparse. Therefore, for each sampled $I$, at most $\hat{d}$
intervals $J$ are within $\hat{\epsilon}$ of $I$.  For each of these we do a diagonal
extension of $I \times J$ to a $w_2$-box $I' \times J'$, and call $\TE(x_{I'},y_{J'},3\hat{\epsilon})$
at cost $O(t(w_2,\hat{\epsilon}))$ for each call. The number of such calls is $O(\frac{\mu(I_{\hat{G}})\hat{d} \log^2 n}{w_2})$.
Using the parameter $\hat{d}=d/\epsilon_i$ in
the $i$th call of the inner iteration of \CA{}, we get
a cost of $O(\frac{\theta n d \log^2 n}{\epsilon_i w_2})t(w_2,\epsilon_i)$ and multiplying
by the $O(1/\theta)$ gives the third summand in the theorem.
\end{proof}

Choosing the parameters to minimize the maximum term in the time bound,
subject to the restrictions of the theorem and using $t(w,\epsilon)=O(\epsilon w^2)$
we have:

\begin{corollary}
\label{cor:run-time}
For all sufficient large $n$, and for $\theta \geq n^{-1/5}$ (both powers of 2)  
choosing $w_1$, $w_2$, and $d$ to be the largest powers of two satisfying:
$ w_1 \le \theta^{-2/7} n^{1/7},$
$
w_2 \le \theta^{1/7} n^{3/7}$,
and 
  $d\le \theta^{3/7} n^{2/7}$, with probability at least $1-n^{-1/7}$, \CA{} runs
in time $\tilde{O}(n^{12/7}\theta^{4/7})$, and outputs the set $\mathcal{R}$ of size at most $\tilde{O}(n^{12/7}\theta^{4/7})$.
\end{corollary}

\begin{proof}
Use the algorithm of \cite{UKK85} that gives $t(w,\epsilon)=O(\epsilon w^2)$.
It is routine to check that these choices satisfy the requirements of Theorem~\ref{thm:time-pIa},
and also that all three terms in the time analysis, and the number of boxes are all
bounded by the claimed bound.
\end{proof}

\section{Min-cost Paths in Shortcut Graphs}
\label{sec:short-path}

We now describe the second phase of our algorithm, which 
uses the set $\mathcal{R}$ output by \CA{} to upper bound $\editd(x,y)$. 
 A {\em shortcut graph} on vertex set $\{0,\ldots,n\} \times \{0,\ldots, n\}$ 
consists of the  H and V edges of cost 1,
together with an arbitrary collection of {\em shortcut} edges
$(i,j)\to (i',j')$ where $i < i'$ and $j < j'$, also denoted by $e_{I,J}$ where $I=\{i,\ldots,i'\}$ and $J=\{j,\ldots,j'\}$, along with their costs.
A  {\em certified graph} (for $x,y$)
is a shortcut graph where every shortcut edge $e_{I,J}$ has cost at least $\editd(x_I,y_J)$.
The min-cost path from $(0,0)$ to $(n,n)$ in a certified graph
upper bounds $\editd(x,y)$.
The second phase algorithm
uses $\mathcal{R}$ to construct a certified graph, and computes
the min-cost  path to upper bound on $\editd(x,y)$.

A certified box $(I \times J,\kappa)$ corresponds to the
 $e_{I,J}$ with cost $\kappa \mu(I)$.
(In the certified graph we use non-normalized costs.) However, the
certified graph built from $\mathcal{R}$ in this way may not have a  path of cost $O(\editd(x,y)+\theta n)$.  We need a modified
conversion of $(I \times J,\kappa)$. If $\kappa \geq 1/2$ we add no shortcut.
Otherwise $(I \times J,\kappa)$ converts to the edge $e_{I,J'}$ with cost $3\kappa\mu(I)$
where $J'$ is obtained by shrinking $J$: $\min(J')=\min(J)+\ell$ and $\max(J')=\max(J')-\ell$
where $\ell=\lfloor \kappa \mu(I) \rfloor$.  By Proposition~\ref{prop:sym diff}, this is a certified
edge.  Call the resulting graph $\Gt$. We claim:

\begin{lemma}
	\label{thm:reduction}
Let $\tau$ be a source-sink path in $G_{x,y}$.
If $\mathcal{R}$ contains a sequence $\sigma$ that $(k,\theta)$-approximates $\tau$ then there is a source-sink path
$\tau'$ in $\Gt$ that consists of the shortcuts corresponding to $\sigma$
together with some H and V edges with
$\cost_\Gt(\tau') \le 5  ( k \cdot \cost_{G_{x,y}}(\tau)  + \theta n).$
\end{lemma}

\begin{proof}
We will modify the path $\tau$ in $G_{x,y}$ to a path $\tau'$ in $\Gt$
of comparable cost.
Let $\{(I_1 \times J_1,\epsilon_1),(I_2 \times J_2,\epsilon_2),\dots, (I_m \times J_m,\epsilon_m)\}$ be the set of certified boxes that $(k,\theta)$-approximates $\tau$. Let $\ell_i= \lfloor \mu(I_i) \cdot \epsilon_i \rfloor$.
Let $L$ be the subset of $[m]$ for which $\epsilon_i \leq 1/2$.  For $i \in L$, let $e_i=e_{I_i,J_i'}$ be the shortcut edge with weight $3\epsilon_i$.  
We claim (1) there is a source-sink path in $\Gt$ that consists of $\{e_i:i \in L\}$ together
with a horizontal path $H_i$ whose projection to the $x$-axis is $I_i$ for each $i \in [m]-L$, and
a collection of (possibly empty) vertical paths $V_0,V_1,\ldots,V_m$ where the $x$-coordinate of $V_i$ for $i >0$ is $\max(I_i)$
and 0 for $V_0$,
and (2) its cost satisfies the bound of the lemma.

For the  first claim, define  for $h \in [m]$, $p_h=(i_h,j_h)$ to be the first point in $\tau_{I_h}$ and define $p_{m+1}=(n,n)$.  We will define
$\tau'$ to pass through all of the $p_h$. In preparation, observe that for $h \in L$, since
 $I_h \times J_h$ $(1-\epsilon_h)$-covers $\tau$, we have $\min(J'_h)=\min(J_h) +\ell_h \geq j_h $ 
and $\max(J'_h) = \max(J)-\ell_h \leq j_{h+1}$. Define the portion $\tau'_h$ between $p_h$ and $p_{h+1}$
by climbing vertically from $p_h$ to $(i_h,\min(J'_h))$ and if $h \in L$ traversing $e_{I_{h},J'_{h}}$ and climbing to $p_{h+1}$
and if $h \not\in L$ then move horizontally from  $(i_h,\min(J'_h))$ to $(i_{h+1},\min(J'_h))$ and then climb to $p_{h+1}$.

For the second claim, we upper bound $\cost(\tau')$.  For $h \in L$, $e_{I_{h},J_h}$ costs $3\ell_h$, and
for $h \not\in L$, the horizontal path that projects to $I_h$ costs $\mu(I_h) \leq 2\ell_h$; the total is at most $\sum_h 3\ell_h$.  
The cost of vertical edges
is $n-\sum_{h \in L} \mu(J_h')=\sum_{h \in L} (\mu(J_h) - \mu(J'_h)) + \sum_{h \not\in L}\mu(J_h) =
\sum_{h \in L} 2\ell_h+\sum_{h \not\in L}\mu(J_h) \leq \sum_h 2\ell_h$, since $\sum_{h}\mu(J_h)=\sum_h \mu(I_h)=n$.  So $\cost(\tau') \leq \sum_h 5\ell_h$. Since
$\sum_{i=1}^m \ell_i \le k \cdot \cost_{G_{x,y}}(\tau) + \theta \cdot n$ by definition of $(k,\theta)$-approximation,
the lemma follows.
\end{proof}

\noindent
{\bf Computing the min-cost.}
We present an $O(n+m\log(mn))$ algorithm to find a min-cost source-sink path in a shortcut graph $\Gt$ with $m$ shortcuts. It's easier to switch to the max-benefit problem: Let $\Ht$ be the same graph with
cost $c_e$ of $e=(i,j)\to (i',j')$ replaced by {\em benefit} $b_e=(i'-i)+(j'-j)-c_e$, (so H and V edges have
benefit 0).
The min-cost path of $\Gt$ is $2n$ minus the max-benefit path of $\Ht$.  
To compute the max-benefit path of $\Ht$,
we use a binary tree data structure with leaves
$\{1,\ldots,n\}$, where each node $v$ stores a number $b_v$, and a collection of lists $L_1$,\ldots,$L_n$, where $L_i$ stores pairs $(e,q(e))$ where the head of $e$ has $x$-coordinate $i$ and $q(e)$ is the max benefit of a path that ends with $e$.

We proceed in $n-1$ rounds. Let the set $A_i$ consist of all the shortcuts whose tail has $x$-coordinate $i$. The preconditions for round $i$ are: (1)  for each leaf $j$, the stored value $b_j$ is the max benefit  path to $(i,j)$ that includes a shortcut 
whose head has $y$-coordinate $j$ (or 0 if there is no such path), (2) for each internal node $v$, $b_v=\max\{b_j:j \text{ is a leaf in the subtree of $v$}\}$,
and (3) for every edge $e=(i',j')\to (i'',j'')$ with $i'<i$, the value $q(e)$
has been computed and $(e,q(e))$ is in list $L_{i''}$.
During round $i$,
for each shortcut $e=(i,j)\to(i',j')$ in $A_i$, $q(e)$ equals
the max of $b_{v}+b_e$ over tree leaves $v$ with  $v \leq j$.  This can be computed in $O(\log n)$ time as max $b_v+b_e$, over $\{j\}$ union
the set of left children of vertices on the root-to-$j$ path
that are not themselves on the path. 
Add $(e,q(e))$ to list $L_{i'}$. After processing $A_i$,
update the binary tree: for each $(e,q(e)) \in L_{i+1}$, let $j$ be the $y$-coordinate
of the head of $e$ and for all vertices $v$ on the root-to-$j$ path, replace $b_v$ by $\max(b_v,q(e))$.
The tree then satisfies the precondition for round $i+1$.
The output of the algorithm is $b_{n}$ at the end of round $n-1$.
It takes $O(n)$ time to set up the data structure, $O(m\log m)$ time to sort the shortcuts, and
$O(\log n)$ processing time per shortcut (computing $q(e)$ and later updating the data structure).

\section{Summing up and speeding up}
\label{sec:sum up}

To summarize,  the algorithm $\UB_{\theta}$ runs CoveringAlgorithm of Section~\ref{sec:covering}, converts the output into a shortcut graph,
and runs the min-cost path algorithm of Section~\ref{sec:short-path}.
By Corollary ~\ref{cor:run-time}, and the quasilinear running time (in the
number of shortcuts) of the min-cost path algorithm, the algorithm $\UB_{\theta}$ runs in time $\tilde{O}(n^{12/7}\theta^{4/7})$.
The construction of the main algorithm {\AED} from {\UB} is standard:

\begin{proof}[Proof of Theorem~\ref{thm:main} from Theorem~\ref{thm:UB}]
Given $\UB_{\theta}$, we construct \AED{}:
Run the aforementioned exact algorithm of \cite{LMS98} with running time $O(n+k^2)$ time on instances of edit distance $k$,
for $O(n+n^{2-2/5})$ time.  If it terminates then it outputs the exact edit distance.
Otherwise, the failure to terminate implies $\editd(x,y) \geq n^{4/5}$.  Now run $\UB_{\theta_j}(x,y)$ for 
$\theta_j=(1/2)^j$ for $j=\{0,\ldots,\frac{\log n}{5}\}$ and output the minimum 
of all upper bounds obtained.  Let $j$ be the largest index with $\theta_j n \geq \editd(x,y)$
(such an index exists since $j=0$ works). The output is at most $\approxfactorB \theta_j n \leq \approxfactorA \editd(x,y)$.
 We run at most $O(\log n)$ iterations, each with running time
$\tilde{O}(n^{2-2/7})$.
\end{proof}

\noindent
{\bf Speeding up the algorithm.}
The running time of {\AED} is dominated by the cost of $\TE($ $z_1,z_2,\epsilon)$
on pairs of strings of length $w \in \{w_1,w_2\}$. We use Ukkonen's algorithm~\cite{UKK85} with $t(w,\epsilon)=O(\epsilon w^2)$.
By replacing Ukkonen's algorithm with $\AED$ we can get a revised algorithm $\AED_1$.
This worsens the approximation factor (roughly multiplying it by the approximation factor of \AED)
but improves running time. 
The internal parameters $w_1,w_2,d$ have to be adjusted to maximize savings. One can iterate this process any constant number of 
times to get faster algorithms with worse (but still constant) approximation factors.
Because of the dependence of the analysis on $\theta$, one does not get a faster edit distance
algorithm for all $\theta \in [0,1]$ but only for $\theta$ close to 1. This approach combined with further ideas is used in \cite{BR20,KS20}.

\paragraph*{Acknowledgements}
We thank the FOCS 2018 committee and especially several anonymous referees
for helpful comments and suggestions.  Michael Saks thanks C. Seshadhri
for insightful discussions on the edit distance problem. The research
leading to these results is partially supported by the Grant Agency of the
Czech Republic under the grant agreement no. 19-27871X, by the
H2020-MSCA-RISE project CoSP-GA no. 823748, and by the European Research
Council under the European Union’s Seventh Framework Programme
(FP/2007-2013)/ERC Grant Agreement no. 616787.
The research is also supported in part by Simons Foundation under grant 332622.

\bibliography{editDistance}


\end{document}